\documentclass[times,doublespace]{qjrms4}

\usepackage{graphicx}
\usepackage[space]{grffile}
\usepackage{latexsym}
\usepackage{textcomp}
\usepackage{longtable}
\usepackage{tabulary}
\usepackage{booktabs,array,multirow}
\usepackage{natbib}
\usepackage{url}
\usepackage{hyperref}
\hypersetup{colorlinks=false,pdfborder={0 0 0}}
\usepackage{etoolbox}
\newif\iflatexml\latexmlfalse

\AtBeginDocument{\DeclareGraphicsExtensions{.pdf,.PDF,.eps,.EPS,.png,.PNG,.tif,.TIF,.jpg,.JPG,.jpeg,.JPEG}}

\usepackage[utf8]{inputenc}
\usepackage[english]{babel}

\usepackage{siunitx}

\iflatexml
\else
\usepackage{amsfonts,amssymb,amsthm} %so that algorithms use line breaks in title
\usepackage[T1]{fontenc}
\usepackage{textcomp}
\usepackage{babel}

\usepackage{tikz}
\usetikzlibrary{shapes,arrows}

\usepackage[parfill]{parskip}
\usepackage{graphicx}
\usepackage{lipsum}

\graphicspath{{figures/}} %Setting the graphicspath
\usepackage{subfig}

\usetikzlibrary{positioning}
\usetikzlibrary{shapes,arrows}

\newcommand{\N}[0]{\mathcal{N}}

\newcommand{\E}[0]{\mathbb{E}}

\newcommand{\R}[0]{\mathbb{R}}

\newcommand{\Kn}[0]{\mathbf{K}_n} %kalman gain
\newcommand{\An}[0]{\mathbf{A}_n} %model linearzation
\newcommand{\Rc}[0]{\mathbf{R}} %obs error covariance
\newcommand{\Sc}[0]{\mathbf{Q}} %model noise covariance
\newcommand{\Pf}[0]{\mathbf{P}^f} %kalman forecast covariance
\newcommand{\Pa}[0]{\mathbf{P}^a} %kalman analysis covariance
\newcommand{\tPf}[0]{\tilde{\mathbf{P}}^f} %kalman forecast covariance
\newcommand{\tPa}[0]{\tilde{\mathbf{P}}^a} %kalman analysis covariance

\newcommand{\jm}[1]{{#1}}

\newcommand{\Lc}[0]{\mathbf{L}}
\renewcommand{\H}[0]{\mathbf{H}}
\newcommand{\Hdag}[0]{\mathbf{H}^\dagger}
\newcommand{\I}[0]{\mathbf{I}}
\renewcommand{\P}[0]{{\Pi}}
\newcommand{\U}[0]{\mathbf{U}}
\newcommand{\T}[0]{\mathbf{T}}
\newcommand{\B}[0]{\mathbf{B}}

\newcommand{\Hq}[0]{\H^q}
\newcommand{\Rq}[0]{\Rc^q}
\newcommand{\Hqp}[0]{\H^{q\perp}}
\newcommand{\Rqp}[0]{\Rc^{q\perp}}

\newtheorem{thm}{Theorem}[section]

\newtheorem{remark}{Remark}

\newtheorem{alg}{Algorithm}

\usepackage{hyperref}

\newcommand\BibTeX{{\rmfamily B\kern-.05em \textsc{i\kern-.025em b}\kern-.08em
T\kern-.1667em\lower.7ex\hbox{E}\kern-.125emX}}

\usepackage{moreverb}

\def\volumeyear{2020}

\begin{document}

%%%
\runningheads{J.M. and E.V.V.}{\emph{Q.~J.~R.
Meteorol. Soc.}}

\title{Particle filters for data assimilation based on reduced order 
data models\footnotemark[2]}

\author{John Maclean\affil{a}\corrauth,\
Erik S. Van Vleck\affil{b}}

\address{\affilnum{a}School of Mathematical Sciences, University of Adelaide, South Australia.
\url{http://www.adelaide.edu.au/directory/john.maclean}\\
\affilnum{b}Department of Mathematics, University of Kansas, USA.
\url{http://people.ku.edu/~erikvv/}}

\corraddr{\url{John.Maclean@adelaide.edu.au}}

%\author{Erik S. Van Vleck}
%\address{Department of Mathematics, University of Kansas, USA.
%\url{http://people.ku.edu/~erikvv/}}
%\author{A.~N.~Other\corrauth}

%\address{John Wiley \& Sons, Ltd, The Atrium, Southern Gate, Chichester, West Sussex, PO19~8SQ, UK}

\begin{abstract}
We introduce a framework for Data Assimilation (DA) in which the data is split into multiple sets corresponding to low-rank projections of the state space. Algorithms are developed that assimilate some or all of the projected data, including an algorithm compatible with any generic DA method.
%\\
The major application explored here is PROJ-PF, a projected Particle Filter.  The PROJ-PF implementation assimilates highly informative but low-dimensional observations. The implementation considered here is based upon using projections corresponding to Assimilation in the Unstable Subspace (AUS).
%\\
In the context of particle filtering, the projected approach mitigates the collapse of particle ensembles in high dimensional DA problems while preserving as much relevant information as possible, as the unstable and neutral modes correspond to the most uncertain model predictions. 
In particular we formulate and numerically implement a projected Optimal Proposal Particle Filter (PROJ-OP-PF) and compare to the standard optimal proposal and to the Ensemble Transform Kalman Filter.

\end{abstract}

%%%
%\paperfield{Data Assimilation}
%\abbrevs{DA, Data Assimilation; PF, Particle Filter; EnKF, Ensemble Kalman Filter; OP-PF, Optimal Proposal PF; ETKF, Ensemble Transform Kalman Filter.}
%\corraddress{}
%\corremail{John.Maclean@adelaide.edu.au}
% \presentadd{Department, Institution, City, State or Province, Postal Code, Country}
%\fundinginfo{JM, ONR, grant: N00014-18-1-2204; ARC, grant DP180100050; EVV, NSF, DMS-1714195 and DMS-1722578}
%\fi

%---

%\papertype{Original Article}

%\title{A new class of assimilation schemes employ reduced order data models, with applications to Particle Filtering}
%\title{Particle filters for data assimilation based on reduced order data models}

%\author[1]{John Maclean}
%\author[2]{Erik S. Van Vleck}

%\affil[1]{School of Mathematical Sciences, University of Adelaide, %South Australia.
%\url{http://www.adelaide.edu.au/directory/john.maclean}}
%\affil[2]{Department of Mathematics, University of Kansas, USA.
%\url{http://people.ku.edu/~erikvv/}}

%\runningauthor{J.M. and E.V.V.}
%%%%

\maketitle

\footnotetext[2]{JM, ONR, grant: N00014-18-1-2204; ARC, grant DP180100050; EVV, NSF, DMS-1714195 and DMS-1722578.}

%\begin{abstract}
%\noindent
%We introduce a framework for Data Assimilation (DA) in which the data is split into multiple sets corresponding to low-rank projections of the state space. Algorithms are developed that assimilate some or all of the projected data, including an algorithm compatible with any generic DA method.
%\\
%The major application explored here is PROJ-PF, a projected Particle Filter.  The PROJ-PF implementation assimilates highly informative but low-dimensional observations. The implementation considered here is based upon using projections corresponding to Assimilation in the Unstable Subspace (AUS).
%\\
%In the context of particle filtering, the projected approach mitigates the collapse of particle ensembles in high dimensional DA problems while preserving as much relevant information as possible, as the unstable and neutral modes correspond to the most uncertain model predictions. 
%In particular we formulate and numerically implement a projected Optimal Proposal Particle Filter (PROJ-OP-PF) and compare to the standard optimal proposal and to the Ensemble Transform Kalman Filter.

%\textbf{Keywords} --- Data Assimilation, Numerical Analysis, Dimension Reduction

%\textbf{AMS} ---- 
%93E11, 93C10, 93B17, 60G35
%\end{abstract}

\section{Introduction} \label{sec:Intro}

%Para1 
Many data assimilation techniques were developed based on extending assumptions of linearity in the phase space and data models and under the assumption of Gaussian errors. Several techniques have proven to be successful in weakening these assumptions, while other techniques have been developed to explicitly overcome these obstacles. Important among these are particle filters \citep{DoucetEtAl2000}, a key subject of this paper. Particle filters have proven to be successful for low dimensional assimilation problems but tend to have difficulty with higher dimensional problems. Different variants of particle filters have been develop to combat these difficulties, including implicit particle filters, proposal density methods, the optimal proposal, etc \citep{Chorin10,SnyderEtAl08,Leeuwen10,Snyder2011}. {Recent work has often focused on the issue of localization (\cite{Farchi18}, e.g.), and two localised particle filtering algorithms \citep{Poterjoy16, Potthast19} have been applied in an operational geophysical framework. The localised particle filter of \cite{Potthast19} contains an element related to the approach taken in this paper. In particular, in \cite{Potthast19} observations are projected onto the subspace spanned by the ensemble of model forecasts. This is shown to effect a significant reduction in the dimension of the data, one which mitigates the issues that high model dimension induces in particle filters.}
%Para2

Our contribution in this paper is to develop a framework for data assimilation schemes in which the data are constrained by an arbitrary projection to lie in some subspace of observation or model space. We explicitly obtain a form for the reduction in data dimension, and an expression that determines how much the posterior of the Bayesian DA scheme is affected by use of the projection. While the projection is not specified, the key idea is that some physically based reduction technique can then be employed in concert with a DA scheme. In such a way the assimilation step is performed in a space of very low dimension. A cognate approach in \cite{Potthast19} projects the data onto the subspace spanned by the forecast ensemble, originating in the Local Ensemble Transform Kalman Filter \citep{Hunt2007}.

The derivation in this paper was motivated in large part by assimilation in the unstable subspace (AUS) techniques. These techniques have largely focused on projecting the phase space model using Lyapunov vectors while employing the original data or observational model. The techniques and framework developed in this paper allow for combinations of (time dependent) projected and unprojected physical and data models, and their formulation is independent of the source of the projections. The framework and techniques lead to several natural applications. In particular we develop, implement, and compare two new particle filter algorithms based upon a dimension reduction technique into the unstable subspace.

\noindent
{We now discuss the historical antecedents of the projections in this manuscript, and connect them to other recent filtering approaches.} The AUS techniques \citep{CaGhTrUb08, TrDiTa10, PaCaTr13, LaSaShSt14, SaSt15} to improve speed and reliability of data assimilation specifically address the partitioning of the tangent space into stable, neutral and unstable subspaces corresponding to Lyapunov vectors associated with negative, zero and positive Lyapunov exponents.
In particular, Trevisan, d'Isidoro \& Talagrand propose a modification of 4DVar, so-called 4DVar-AUS, in which corrections are applied only in the unstable and neutral subspaces \citep{TrDiTa10,PaCaTr13}. %Much of the progress in AUS techniques has been in analysing AUS Extended Kalman Filter techniques.
These techniques are based on updating in the unstable portion of the tangent space and may be interpreted in terms of projecting covariance matrices during the assimilation step.

%On the other hand, the stable subspace may also be exploited. 
Motivated by these techniques for assimilation in the unstable subspace, in \cite{ProjShadDA} a new method is developed for data assimilation that utilizes distinct treatments of the dynamics in the stable and non-stable directions. 
In particular, the first phase of this development has involved employing time dependent Lyapunov vectors
to form a subsystem with tangent space dynamics similar to the unstable subspace of the original state space model. 
This was motivated by AUS techniques. The key piece of \cite{ProjShadDA} related to this work is the following projected model update. For a smooth discrete time model $u_{n+1} = F_n(u_n)$ and projection  $\P_n$, %alternate solving the two projected subsystems, using shadowing refinement in the first subsystem and forward in time evolution in the second. 
%Here $k$ is the Newton iterate.
%%%
and for $\{u_n^{(0)}\}_{n=0}^N$ any reference solution, %(for example after being updated in one of the subsystems), 
solve for $\{d_n\}_{n=0}^N$:%, $d_n := \P_n\delta_n$: %and/or $\{e_n\}_{n=0}^N$, $e_n := (\I-\P_n)\delta_n$, in their respective subsystems:
\begin{equation}
\label{PIPsys}
u_{n+1}^{(0)} + d_{n+1} = \P_{n+1}F_n(u_n^{(0)}+d_n),\,\,\,
%u_{n+1}^{(0)} + e_{n+1} = (\I-\P_{n+1})F_n(u_n^{(0)}+e_n),\,\,\, 
n=0,...,N-1.
\end{equation}
This manuscript develops a complementary approach to project the data model. The two approaches will be compared in Section~\ref{sec:comp}.

Another branch of projected DA schemes use the `Dynamically Orthogonal' (DO) formulation \citep{SL09,Sapsis}, in which the forecast model 
is broken into a partial differential equation governing the mean field and a number of stochastic differential equations describing the evolution of components in a time-dependent stochastic subspace of the original differential equation. The DO approach was used to assimilate with different DA schemes in the subspace and mean field space in \cite{SL13,MajdaQiSapsis14,QiMajda15}. These techniques use both a projected and mean field model to make a forecast, similar to using \eqref{PIPsys}. 
The data is naturally split into the `projected' and remaining components without using a projected data model (see, e.g., \eqref{pdata2}) explicitly, as the data may be confined to the DO-subspace by simply subtracting the forecast mean field. This attractive feature of the DO methodology bypasses the need to derive a projected data model, as the covariance structure of the data does not change. 

Projection-based DA schemes have been developed to assimilate coherent structures \citep{MSJ17} or features \citep{Morzfeld18} in the data. These approaches have used likelihood-free sequential Monte Carlo methods, or an ad hoc `perturbed observations' approach, to deal with the difficulty of calculating the likelihood function for a coherent structure. The derivation in this paper may lead to an explicit likelihood for data-derived coherent structures/features obtained via a projection.
Additional sources for projections may be found in the review of projection based model reduction techniques \citep{BGW15}.

We develop a projected DA framework and algorithms for arbitrary time dependent orthogonal  projections, but are mainly interested in the AUS approach where the projections identify the unstable/neutral subspace. To determine the projections we will employ standard techniques for approximation of Lyapunov exponents, e.g., the so-called discrete QR algorithm (see \cite{DVV07,DiVV15}).

Unlike most past work related to AUS our primary focus is on developing a systematic approach to confining the data, not the model, to the unstable subspace. In some
of the initial works on AUS \citep{CTU07,CaTrDeTaUb08}, either target observations at the location
where the unstable mode attains its maximum value, or only the observations falling in the vicinity of the maximum, were assimilated. Albeit empirical, that choice already signified using only data projected on an approximation of the unstable subspace, that was obtained by Breeding on the
Data Assimilation Cycle (BDAS). Furthermore, \cite{ToHu13,Bocquet17,Grud2018b,FZ2018} are all at least in part devoted to discussing the necessary and/or sufficient criteria for filter stability in terms of the projection of the observations into the unstable/neutral/weakly stable directions and this is directly related to the choice of adaptive observation operators in \cite{LaSaShSt14}.

%For comparison, past work on projected DA algorithms is now reviewed in the above framework. The standard AUS approach (\cite{PaCaTr13}, e.g.) can be derived by calculating the assimilation step using the projected model \eqref{Psys} and the original, unprojected data \eqref{data}. This derivation will be performed in the next section. The projected shadowing approach of \cite{ProjShadDA} iterates between assimilating unprojected data \eqref{data} into the projected model \eqref{Psys} and evolving the complementary orthogonal forecast model \eqref{IPsys} without assimilation.

If the non-stable subspace is relatively low dimensional this makes applications of techniques such as particle filters appealing. 
Particle Filters \citep{doucet} are particularly effective for nonlinear problems and for the tracking of non-Gaussian, multi-modal probability distributions; but they suffer from the so-called curse of dimensionality \citep[e.g.]{SnyderEtAl08, Snyder2011,MTAC12,van2012particle}. There is a known formulation that minimises degeneracy; however, even with a linear model, it is known that the computational cost of this "Optimal Proposal" Particle Filter scales like the exponential of the observation dimension \citep{Snyder2011}. Our aim in this work is to avoid the established limits of particle filter performance by reducing the observation dimension in a sensible way. Other efforts to bypass this limitation include, e.g., the Equivalent Weights Particle Filter \citep{Leeuwen10}. One attractive feature of our approach is that it is a reformulation of the standard DA problem rather than a specific algorithm, and so it is compatible with these advanced particle filters.

The availability of the projection into the unstable subspace will also allow us to develop a novel approach to resampling. It is necessary to periodically refresh any particle ensemble, and some noise is usually added at this step. To avoid forcing the ensemble off the attractor with this added noise we confine most of the noise to the unstable subspace, which improves the filter accuracy and reduces the incidence of resampling in later steps. 
%The projections will be used to construct projected observation operators with lower rank or dimension while simultaneously updating the likelihood as well.

% The projected state space and data models we develop will be applied in several different contexts. These include the development of a projected particle filter or PF-AUS (Particle Filter for Assimilation in the Unstable Subspace), assimilation of coherent structure or features \cite{MSJ17}, Blended Particle Filtering, and  PF-AUS with the Optimal Proposal. These application have in common the decomposition of the state space and/or observation space and the use of known techniques in the context of projected space.

This paper is organized as follows. Data assimilation is reviewed in section~\ref{sec:da} and projected DA is 
formulated in section~\ref{sec:pfaus}. Algorithms for using the new projected data are introduced (section~\ref{sec:PDA}) and applied to AUS with several numerical experiments (section~\ref{sec:AUS}). A discussion (section~\ref{sec:disc}) and bibliography conclude the paper.

%%%%%%%%%%%%%%%%%%%%%%%%%%%%%%%%%%%%%%%%%%%%%%%%%%%%%%%%%%%%%%%%%%%%%%%
\section{Data Assimilation} \label{sec:da}
Data assimilation methods combine orbits from a dynamical system model with measurement data to obtain an improved estimate for the state of a physical system.   
In this paper we develop a data assimilation method in the context of the discrete time stochastic model
\begin{align} \label{model}
u_{n+1} = F_n(u_n) + \sigma_n,\,\,\, n=0,1, ...
\end{align}
where $u_n\in\R^N$ are the state variables at time $n$
and $\sigma_n \sim \N(0,\Sc)$, i.e., drawn from a normal distribution with mean zero and model error covariance $\Sc$.
Let the sequence
$\{u_0^t, u_1^t, \dots \}$,
be a distinguished orbit of this system,
 referred to as the \emph{true solution} of the model, and presumed to be unknown.
As each time $t_n$ is reached we collect an observation $y_n$  related to $u_n^t$ via 
\begin{equation}\label{data}
	y_n = \H u_n^t + \eta_n, \qquad y_n \in \R^M
\end{equation}
where $\H:\R^N \rightarrow \R^M$, $M\le N$, is the observation operator, and the noise variables $\eta_n$ are drawn from a normal distribution $\eta_n \sim \N(0,\Rc)$ with zero mean and known observational error covariance matrix $\Rc$.  In general the observation operator can be nonlinear.\\
%Data assimilation can be viewed as the problem of finding an orbit or approximate orbit $\B{u}=\{ u_0, u_1, \dots, u_N\}$, $u_n \in \R^N$, of the model \eqref{model}, such that the differences $\| y_n - H u_n\|$, $n=0,\dots$ are small in an appropriately defined sense. This is done with the aim of minimizing the unknown error $\| u_n - u_n^t \|$.
We formulate DA under the ubiquitous Bayesian approach. Consider the assimilation of a single observation, $y_n$, at time step $n$. Given a prior estimate $p(u_n)$ of the state, Bayes' Law gives
\begin{align*}
p(u_n|y_n)\,  &{ \propto }\, p(y_n|u_n)p(u_n),
\end{align*}
Using \eqref{data} the likelihood function is, up to a normalization constant,
\begin{align}
\label{likelihood}
p(y_n|u_n) \propto \exp\left[-\frac{1}{2} \left(y_n - \H u_n\right)^T \Rc^{-1} \left(y_n - \H u_n\right) \right] \;.
\end{align}
This procedure, which we have written for the assimilation of data at a single observation time, readily extends to the sequential assimilation of observations at multiple times under the assumptions that the state is Markovian and the observations at different times are conditionally independent (see for example \cite{HandbookDA}).\\

In the following we introduce some key DA schemes. Not much detail is given here, but the interested reader is referred in particular to three recent books on DA, \citep{ReichCotter15,Law2015,Asch16}.
\subsection{Kalman Filtering} \label{sec:KF}
The Kalman Filter and later extensions are ubiquitous in DA, and are now briefly described. For a linear model, i.e. where \eqref{model} is
\begin{align} \label{linmodel}
u_{n+1} = \An u_n + \sigma_n,
\end{align}
and for the linear observation operator $\H$, the Kalman Filter calculates the exact posterior ${u_n|y_n \sim \N(u^a_n,\Pa_n)}$, where the \emph{analysis} variables are
\begin{align}
\label{ua} u^a_n =& u^f_n + \Kn(y_n - \H u^f_n)\;,
\\
 \label{P}   \Pa_n =& \left(\I-\Kn\H\right)\Pf_n.
\end{align}
The weight matrix $\Kn$ is the Kalman gain matrix
\begin{align}
    \label{K}
    \Kn = \Pf_n \H^T \left(\H \Pf_n \H^T + \Rc\right)^{-1} \;.
\end{align}
The superscript $f$ is reserved for \emph{forecast} variables, obtained at time $n$ by using \eqref{linmodel} to update $\{u^a_{n-1},\Pa_{n-1}\}$, 
$$u^f_n = \mathbf{A}_{n-1} u^a_{n-1} + \sigma_{n-1}\;,$$
$$\Pf_n = \mathbf{A}_{n-1} \Pa_{n-1} \mathbf{A}_{n-1}^T + \Sc\;.$$
%The covariance matrices $\Pf_n$ and $\Pa_n$ are not bolded in order to distinguish them from the projections $\P_n$ that are the focus of the remainder of the paper.
%\\
Two extensions of the Kalman Filter are prevalent in nonlinear DA, the Extended Kalman Filter (EKF) and Ensemble Kalman Filter (EnKF). Neither give the exact posterior for a nonlinear model.

\subsubsection{Extended Kalman Filter} \label{sec:EKF}
The nonlinear model \eqref{model} is used to make the forecast $u^f_n$, and then the Kalman Filter update is applied using the linearisation 
$$\An = \left. \frac{\partial F_n}{\partial u}\right|_{u^a_{n}} \;.$$
If the observation operator is a nonlinear function $h()$, the linearization
$$\H_n = \left. \frac{d h}{d u}\right|_{u^f_n}$$
is used everywhere except to compute the \emph{innovation} $y_n - h(u^f_n)$ in the calculation of $y^a_n$.\\
The EKF is suitable for low dimensional nonlinear filtering, but the required linearizations are nontrivial for high-dimensional filtering. The EnKF by contrast is well suited to high dimensions.

\subsubsection{Ensemble Kalman Filter} \label{sec:EnKF}

The Ensemble Kalman Filter is a Monte Carlo approximation of the Kalman Filter that is well suited to high dimensional filtering problems, introduced in \cite{Evensen94,Burgers98}. An ensemble of forecasts $u_n^{f,i}$ are made at time $t_n$, $i$ from $1$ to $L$. Then the forecast covariance $\Pf_n$ is approximated by the sample covariance of the ensemble, and the analysis ensemble $u_n^{a,i}$ is obtained in such a way that its mean $\bar{u}_n^a = \frac{1}{L}\sum_i u_n^{a,i}$ satisfies \eqref{ua} and its sample covariance satisfies \eqref{P}. In this paper we will use analysis updates corresponding to the Ensemble Transform Kalman Filter (ETKF) \citep{Bishop2001}.% or Local Ensemble Transform Filter (LETKF) \cite{Hunt2007}.
For more details and a modern introduction to the Ensemble Kalman Filter, see e.g. \cite{evensen2009data}.

\subsection{The Particle Filter} \label{sec:PF}
Particle Filters (PF) are a collection of particle based data assimilation schemes that do not rely on linearization of the dynamics or Gaussian representations of the posterior; see \cite{doucet} for a comprehensive review. The basic idea is to represent the prior distribution $p(u_n)$, previously the forecast, and the posterior distribution $p(u_n|y_n)$, previously the analysis, by discrete
probability measures. Suppose that at time $n-1$ we have the posterior distribution $(u_{n-1}^i,w_{n-1}^i)$, supported on points $u_{n-1}^1, \ldots u_{n-1}^L$ and with
 weights $w_{n-1}^1, \ldots w_{n-1}^L$. Each $w_{n-1}^i \ge 0$ and $\sum_{i=1}^L w_{n-1}^i=1$.
Here L is the number of particles that are used to approximate the distribution $\Pi^{}_{n-1}$. The two key steps in the Particle Filter are as follows:\\
\noindent {\em Prediction step.}  Propagate each of the particles $u_{n-1}^i \mapsto  u_{n}^i$. One simple choice, the bootstrap PF, is to use the state dynamics \eqref{model} to forecast each particle. \\
This gives the forecast probability distribution as a discrete probability measure concentrated on
$L$ points $\{ u_n^i\}_{i=1}^L$ with weights $\{ w_{n-1}^i\}_{i=1}^L$.

\noindent {\em Filtering step.} 
Update the weights $\{ w_{n-1}^i\}_{i=1}^L$ using the observation $y_n$. In the bootstrap PF the update is
$$w_n^i = c\, w_{n-1}^i p(y_n| u^i_n),$$ where $c$ is chosen so that $\sum_{i=1}^L w_{n}^i=1$.

This scheme is  easy to implement but suffers from severe degeneracy, especially in  high dimensions. That is, after a few time steps all the weight tends to concentrate on a few particles. A common remedy is to monitor the Effective Sample Size (ESS) and resample when the ESS drops below some threshold in order to refresh the particle cloud; see e.g. \cite{doucet,HandbookDA}.

\subsubsection{The Optimal Proposal} \label{sec:op}
The optimal proposal particle filter (OP-PF) \citep{SnyderEtAl08,DoucetEtAl2000,Snyder2011,van2012particle} %falls within the class of methods known as proposal transition distribution methods 
attempts to address the degeneracy issue in particle filters with the aim of ensuring that all posterior particles have similar weights. The `proposal' is the distribution used to update the particles from one time step to the next.
In the prediction step in the basic particle filter above, the particles are updated using the model, so the proposal density in that approach is (compare \eqref{model}) ${p(u^i_n | u^i_{n-1}) \sim \N(F_{n-1}(u^i_{n-1}),\Sc)}$. 

The optimal proposal density is ${p(u^i_{n}|u^i_{n-1},y_{n})}$. Given the additive noise of the model \eqref{model}, the optimal proposal update in each particle is Gaussian with ${p(u^i_{n}|u^i_{n-1},y_{n})\sim\N(m^i_{n},\Sc_{p})}$, where
\begin{align}
\label{opS} \Sc^{-1}_{p} =& \Sc^{-1} +  \H^T  \Rc^{-1}  \H  \;, \\
% \label{opM} m_{n+1} =& \Sigma_{n+1}\left(\Sigma^{-1}_{u,n}F_n(u_n) + \P_\H \U_nn  \Sigma^{-1}_{y,n+1} y_{n+1}^q\right)\;.
\label{opM} m_{n}^i =& F_{n-1}(u_{n-1}^i) + \Sc_{p}\H^T \Rc^{-1}\left(y_{n} - \H F_{n-1}(u_{n-1}^i)\right)\;.
\end{align}
The prefactor $\Sc_{p}\H^T \Rc^{-1}$ can be written as the Kalman gain \eqref{K} (albeit with $\Pf_n=\Sc$) by an application of the Sherman-Morrison-Woodbury formula (see e.g. \cite{Kalnay03}, p. 171).

Two applications of Bayes' law (e.g. in \cite{Snyder2011}) show that the weight update for the $i$-th particle drawn from this proposal satisfies ${w_{n}^i \propto\; p(y_{n}|u_{n-1}^i)w_{n-1}^i}$ and is also Gaussian,
\begin{align}
\label{opW} w_n^i \propto& \exp\left[-\frac{1}{2}(I_n^i)^T\left(\H\Sc\H^T + \Rc \right)^{-1}(I_n^i)\right] w_{n-1}^i\, .
\end{align}
where $I_n^i := y_{n} - \H F_{n-1}(u_{n-1}^i)$.

As mentioned in the previous section, degeneracy - characterised by a single particle with weight of approximately $1$ - is a common problem in the PF. In \cite{Snyder15} it is shown that, of all PF schemes that obtain $u_n^i$ using $u_{n-1}^i$ and $y_n$, the `optimal proposal' above has the minimum variance in the weights. That is, it suffers the least from weight degeneracy. In \cite{VanLeeuwen18} this result is extended to any PF scheme that obtains $u_n^i$ using $i$, $u_{n-1}^{1:L}$ and $y_n$. \\
The distributions required to apply the Optimal Proposal are not always available (the additive model error of \eqref{model} and linear observation operator of \eqref{data} are used above to obtain closed forms for the individual particle updates and weight updates), but when OP-PF can be formulated it is the least degenerate of a large class of filters. However, in \cite{Snyder2011} it is shown that the optimal proposal requires an ensemble size L satisfying $\log L \propto N\!\times\! M$ for a linear model, or will suffer from filter degeneracy.
That is, filter degeneracy is intimately connected to model and observation dimension, and is a fundamental obstacle to Particle Filtering in high dimensional problems.

%%%%%%%%%%%%%%%%%%%%%%%%%%%%%%%%%%%%%%%%%%%%%%%%%%%%%%%%%%%%%%%%%%%%%%%
\section{Projected Data Models} \label{sec:pfaus}

We now develop an approach to decompose the observations using projections defined in state space. A wealth of techniques from dynamical systems theory can then be used to obtain low-dimensional data models. \\
%time dependent stability properties of the state space model. 
% For example, for state space/data models of the form (\ref{model})/(\ref{data}),
% the corresponding state/data models in the unstable subspace
% \begin{align}
% \label{p.statedata} 
% & u_{n+1}^{(0)} + \P_{n+1} u_{n+1} = \P_{n+1}F_n(u_n^{(0)}+\P_n u_n)+\P_{n+1}\sigma_n,\\
% & \tilde y_{n+1} \equiv \P_{n+1}^\H  \Hdag y_{n+1} = \P_{n+1}^\H  u_{n+1} + \P_{n+1}^\H  \Hdag\eta_{n+1},
%         \end{align}
% %\begin{align}
% %\label{p.data} \P_{n+1}^\H  \Hdag y_{n+1} = \P_{n+1}^\H  u_{n+1} + \P_{n+1}^\H  \Hdag \eta_{n+1},
% %\end{align}
% where $\Hdag = \H^T (\H\H^T)^{-1}$ (the pseudo inverse of $\H$, assuming $\H$ is linear and full rank), and
% $\P_{n+1}^\H $ is the orthogonal projection onto the intersection of the subspaces spanned by the columns of $\U_{n+1}$ and the rows of $\H$.
% Note that $\P_\H = \H^T(\H\H^T)^{-1}\H = \Hdag \H$. Here (compare with (\ref{model}) and (\ref{data})) $\P_{n+1}\sigma_n \sim \N(0,\P_{n+1}\B\P_{n+1})$ and $\P_{n+1}^\H  \Hdag \eta_{n+1} \sim \N(0,\P_{n+1}\Hdag\Rc(\Hdag)^T\P_{n+1})$.
% Dimension reduction is obtained since
% $x_{n+1}\equiv \P_{n+1}\delta_{n+1}$ and $\P_{n+1}^\H  \Hdag y_{n+1}$ can be
% written as a linear combination of the $p<d$ columns of $\U_{n+1}$.
Suppose that at time $n$ a dynamically significant rank $p$ orthogonal projection $\P_n \in\R^{N\times N}$ is available, as well as data $y_n \in \R^M$. %, along with a splitting $\P_n = \U_n\U_n^T$, where $\U_n \in \R^{N\times p}$ has orthonormal columns. 

%In the application motivated above the projection identifies unstable modes, but the derivation below does not require this interpretation of $\P_n$.

%The idea is to formulate a framework for DA that uses the standard model dynamics \eqref{model} and collects observations through \eqref{data}, but which performs the data assimilation step in the subspace of model space spanned by the projection $\P_n$. 
The main result will be to define a projected observation $y^q_n \in \R^p$, and derive a corresponding data model 
\begin{equation}\label{projdata} 
y^q_n = \U_n^T \P_\H u^t_n +  \gamma_n
\end{equation}
that is a linear transformation of \eqref{data}, where 
%$\U_n$ is derived from $\P_n$, $\P_\H$ from $\H$, 
$\P_n$ and $\P_\H$ are orthogonal projections with $\P_n = \U_n \U_n^T$ ($\U_n^T \U_n = I$)
and $\P_\H = \H^T (\H\H^T)^{-1}\H$, 
and $\gamma_n$ has known distribution. The projected data contains only the components of observations that can be written as a linear combination of the columns of $\U_n$. It can be used in place of the original data in any DA algorithm, or used in concert with the original data in the novel Particle Filtering algorithm developed in Section~\ref{sec:PDA}.

We will derive \eqref{projdata} in the following three steps.

\subsection*{Step One: lift the data into model space}
In order to apply the projection $\P_n$ to data, we first need to find an equivalent representation of the data in model space. \\
Assuming $\H$ has full row rank, we define an $N$-dimensional vector $\tilde y_n = \Hdag y_n$ where $\Hdag = \H^T(\H\H^T)^{-1}$. The data model for $\tilde y_n$ is
\begin{align*}\tilde y_n =& \Hdag y_n\\
=&  \P_\H u^t_n + \Hdag\eta_n \\
=& \P_\H u^t_n + \psi_n 
\end{align*}
where $\P_\H = \Hdag \H$ is an orthogonal projection, and $\psi_n \sim \N(0,\Hdag \Rc(\Hdag)^T)$. \\

\noindent Using that $\H \Hdag = \I$ one readily confirms that $\H \tilde y_n = y_n = \H u^t_n + \eta_n$. That is, the observation operator collapses $\tilde y_n$ onto the standard data model. The transformation through $\Hdag$ has not affected the output of a DA scheme, as $p(\tilde y_n|x) = p(y_n|x)$; however $\tilde y_n$ is of compatible dimension with $\P_n$.

% Using the new data model and the statistics of $\psi$ we see the likelihood is
%\begin{align*}
%p(\tilde y|u) \propto& \exp\left(-\frac{1}{2}\left(\tilde y-\P_\H u\right)^T \left(\Hdag \Rc(\Hdag)^T\right)^{-1}\left(\tilde y-\P_\H u\right)\right) \\
%=&\exp\left(-\frac{1}{2}\left(\tilde y-\P_\H u\right)^T \H^T R^{-1} \H\left(\tilde y-\P_\H u\right)\right).
%\end{align*}
%
%\begin{remark}
%This construction does not change the weight update of a particle filter: that is, $p(\tilde y|u) = p(y|u)$. We have
%
%	\begin{align*}
%	p(\tilde y|u) =& \left(\tilde y - \P_\H u\right)^TH^T R^{-1} H \left(\tilde y - \P_\H u\right)\\
%	=& \left(\H^T(\H\H^T)^{-1}( y - \H u)\right)^T \H^T\Rc^{-1}\H \left(\H^T(\H\H^T)^{-1}( y - \H u)\right) \\
%	=& \left( y - \H u\right)^T\left(\H^T(\H\H^T)^{-1}\right)^T \H^T\Rc^{-1}\H \H^T(\H\H^T)^{-1}( y - \H u) \\
%	=& \left( y - \H u\right)^T(\H\H^T)^{-1} \H \H^T\Rc^{-1}\H \H^T(\H\H^T)^{-1}( y - \H u)\\
%	=& \left( y - \H u\right)^T\Rc^{-1}( y - \H u)\\
%	=&p(y|u) \;.
%	\end{align*}
%\end{remark}

\subsection*{Step Two: project the data into a rank $p$ subspace} 
We now make use of the orthogonal projection $\P_n$. The idea is to formulate a new data model, along the lines of $\P_n \tilde y_n = \P_n \P_\H u^t_n + \P_n \Hdag \eta_n$, that contains only the components of the observation that align with the projection. The projected data models that are developed here may be considered as generalizations of the construction of observation operators (see \cite{Grud2018b} Def. 13 and \cite{LaSaShSt14}).

\noindent

Define $y_n^p = \P_n  \tilde y_n = \P_n  \Hdag y_n \in \R^N$, the projected observation. The data model is  
\begin{align}
\nonumber y^p_n =& \P_n  \Hdag y_n\\
\label{pdata1} =&  \P_n  \P_\H u^t_n + \xi_n %\\
%  \equiv & \H_n^p  u^t_n + \xi_n
\end{align}
where $\xi_n\sim\N(0,\P_n  \Hdag \Rc(\Hdag)^T \P_n)$. The data model $y^p_n$ has a singular normal distribution with support in the $p$-dimensional subspace of model space spanned by the projection $\P_n$, and the likelihood of this distribution can be written using the pseudo-inverse (see e.g. \cite{TK15}) as 
\begin{align}
\label{tmp.p} p( y^p_n|u) \propto& \exp\left(-\frac{1}{2} (I_n^p)^T \left(\P_n  \Hdag \Rc(\Hdag)^T \P_n \right)^{\dagger} I_n^p\right)
%\left( y^p_n-\P_n \P_\H u\right)\right) \;. 
\end{align}
where $I_n^p := y^p_n-\P_n \P_\H u$.

\begin{remark}
The product $\P_n\P_\H$ is not generally an orthogonal projection, and in some circumstances it might be desired to instead identify the projection $\P_n^\H$ that is the intersection of $\P_n$ and $\P_\H$. This projection $\P_n^\H$ may be approximated by Von Neumann's algorithm or Dykstra's projection algorithm; see Appendix~\ref{dpa} for a review. The projection $\P_n^\H$ should only be used if the transversality condition $p+M-N>0$ is satisfied; otherwise there is no guarantee of any intersection between $\P_n$ and $\P_\H$.\\

% \noindent For convenience we define the projection $\P_n$ as
% \begin{equation}\label{ProjDef}
% \P_n = \left\{
% \begin{array}{lr}
% \P_n\, ,& p+M-N\leq 0,\cr
% \P_n^\H\, ,& p+M-N>0. 
% \end{array}
% \right.
% \end{equation}
\end{remark}

\subsection*{Step Three: reduce the projected data to a $p$-vector}

To make explicit the reduction in the data dimension that has been obtained by $y^p_n$ we introduce a low dimensional data model. 
Denote by $\U_n$ the matrix with orthonormal columns satisfying $\P_n = \U_n\U_n^T$. This matrix may be already known (in the examples in Section~\ref{sec:pfaus} $\U_n$ is obtained first, and then $\P_n$ is calculated from $\U_n \U_n^T$), or $\U_n$ may be found via the singular value or Schur decompositions. For the case $\P_n = \P_n^\H$ we redefine $p$ as the rank of $\P_n$. \\
Define $y^q_n = \U_n^T y^p_n \equiv \U_n^T \tilde y_n \in \R^p$, with the associated data model
\begin{align} \label{pdata2}
y^q_n = \Hq_n u^t_n +  \gamma_n\;,
\end{align}
where $\Hq_n = \U_n^T \P_\H$, $\gamma_n \sim \N(0,\Rq_n)$, and $\Rq_n =     \U_n^T \Hdag \Rc (\Hdag)^T \U_n$. 

The transformations between and dimensions of the different data variables defined in this section are illustrated in Figure~\ref{fig:guide}.

%Draw tikz flowchart of data types

% Define block styles
\tikzstyle{decision} = [diamond, draw, fill=blue!20, 
    text width=4.5em, text badly centered, node distance=3cm, inner sep=0pt]
\tikzstyle{block} = [rectangle, draw, fill=blue!20, 
    text width=3em, text centered, rounded corners, minimum height=3em]
\tikzstyle{line} = [draw, -latex']
\tikzstyle{cloud} = [draw, ellipse,fill=red!20, node distance=4em,
    minimum height=2em]
    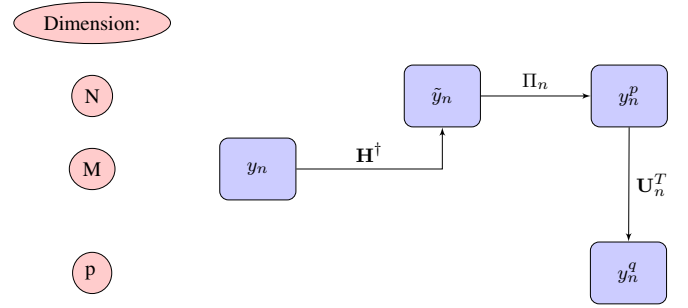
\begin{figure}
    \resizebox{0.49\textwidth}{!}{
\begin{tikzpicture}[node distance = 2cm, auto]
    % Place nodes
    \node [cloud] (dim) {Dimension:};
    \node [cloud, below=1.5em of dim] (model) {N};
    \node [cloud, below=1.5em of model] (obs) {M};
    \node [cloud, below=3em of obs] (proj) {p\,};
    \node [block, right=5em of obs] (yn) {$y_n$};
    \node [block, right=14em of model] (tyn) {$\tilde y_n$};
    \node [block, right=23em of model] (pyn) {$y^p_n$};
    \node [block, right=23em of proj] (qyn) {$y^q_n$};
    \path [line] (yn) -| node [near start] {$\Hdag$} (tyn);
    \path [line] (tyn) -- node {$\P_n$} (pyn);
    \path [line] (pyn) -- node {$\U_n^T$} (qyn);
\end{tikzpicture}}
\caption{The progression from the original data $y_n$ to low-dimensional, projected data $y^q_n$. The rectangular boxes contain data, or data-derived constructs. The height of the box shows the dimension of the data at each step. Note that in practice one does not need to compute $\tilde y_n$ or $y^p_n$.}
\label{fig:guide}
\end{figure}
%tikz end

\subsection{Properties of the projected data}
% The following result establishes that the data model \eqref{pdata2} and \eqref{pdata1} are equivalent.
\noindent
\begin{thm}[Equivalence of $y^p_n$ and $y^q_n$]\label{thm:main}
For the data models associated with $y^p_n$ and $y^q_n$ given by
(\ref{pdata1}) and (\ref{pdata2}), respectively, $p(y^q_n|u) = p(y^p_n|u)$. \\
%If an additional condition, that $p\le M$ if $\P_n=\P_n$ or $p+M-N>0$ if $\P_n=\P_n^\H$, is satisfied
%then the covariance matrix $\U_n^T \Hdag \Rc (\Hdag)^T \U_n$ of $y^q$ is invertible and (equivalently) $y^q$ has a standard normal distribution.
\end{thm}
\begin{proof}
The matrix $\U_n$ has orthonormal columns, so $\U_n^\dag = \U_n^T$ and for any matrix $\B$
\begin{align*}
\left(\U_n \B\right)^\dag =& \B^\dag \U_n^\dag = \B^\dag \U_n^T \;,\\
\left( \B\U_n^T\right)^\dagger =& (\U_n^T)^\dag \B^\dagger  = \U_n \B^\dag  \;.
\end{align*}
Applying these results to \eqref{tmp.p}, and using that $\P_n = \U_n\U_n^T$, $y^p_n = \U_n y^q_n$, $\U_n^T\U_n = \I$, $I_n^p:= y^p_n-\P_n \P_\H u$, and $I_n^q:= y^q_n-\U_n^T \P_\H  u$,
\begin{align*}
p(y^p_n|u) \propto& \exp\left(-\frac{1}{2}(I_n^p)^T \left(\P_n  \Hdag \Rc(\Hdag)^T \P_n \right)^{\dagger} I_n^p\right)  
\\
=& \exp\left(-\frac{1}{2}\left(\U_n I_n^q\right)^T %\left(\U_n\U_n^T  \Hdag \Rc(\Hdag)^T\U_n\U_n^T\right)^{\dagger}
(\tilde \Rc_n)^\dagger
\left(\U_n I_n^q\right)\right) 
\\
=& \exp\left(-\frac{1}{2}(I_n^q)^T \U_n^T \U_n \left(\Rq_n\right)^{\dagger}\U_n^T\U_n I_n^q\right) 
\\
=& \exp\left(-\frac{1}{2}(I_n^q)^T \left(\Rq_n\right)^{\dagger} I_n^q\right) 
\\
=& p(y^q_n|u) 
%\;.
\end{align*}
where $\tilde \Rc_n := \U_n\U_n^T  \Hdag \Rc(\Hdag)^T \U_n\U_n^T.$
%\note{still need to prove conditions on the invertibility of the covariance. Or weaken the theorem statement and note the invertibility afterwards.}
\end{proof}
If in addition $p\le M$ (or $0<p+M-N\leq M$ for $\P_n\equiv\P_n^\H$), and if $\H\U_n$ is full rank, 
then the covariance matrix $\Rq_n$ of $y^q_n$ is invertible and $y^q$ has a standard normal distribution.
More generally for $(\H\H^T)^{-1}\Rc (\H\H^T)^{-1}= \Lc^T \Lc$, the Cholesky factorization, consider the SVD of $\Lc\H\U_n = \mathbf{S}\mathbf{\Sigma} \mathbf{V}^T$. The rank of
the covariance matrix $\Rq_n = \U_n^T \Hdag \Rc (\Hdag)^T \U_n = \mathbf{V} \mathbf{\Sigma}^T \mathbf{\Sigma} \mathbf{V}^T$ is equal to the number of non-zero singular values of $\mathbf{\Sigma}$. 
%Regardless, $(\U_n^T \Hdag \Rc (\Hdag)^T \U_n)^\dagger = \U_n^T \H^T \Rc^{-1} \H\U_n$.

Theorem~\ref{thm:main} provides a blueprint for any DA scheme to be efficiently implemented with projected observations, involving the following changes: the observation $y_n$ is replaced with $y^q_n$, the observation operator $\H$ is replaced with $\Hq_n$, and the assumed measurement covariance $\Rc$ is replaced with $\Rq_n$.

\subsection{The orthogonal data model}
Though the focus of this paper is on the projected data, a data model for the complementary orthogonal projection $\I-\P_n$ is easy to write down. Define 
\begin{align}\label{IPdata}
y^{q\perp}_n = \left(\U_n^\perp\right)^T \tilde y_n\in\R^{N-p}\;,
\end{align}
where $\U_n^\perp (\U_n^\perp)^T = \I-\P_n$. The two projected data models are not independent in general and have joint distribution
\begin{align}\label{datacov}
\left[\begin{gathered}y^q_n\\ y^{q\perp}_n\end{gathered}\right] \sim \N\left(
\left[\begin{gathered}\Hq_n\, u^t_n\\ \Hqp_n u^t_n\end{gathered}\right],\quad
\left[\begin{gathered}
\Rq_n      \qquad
\Rc^{q}_{12,n}\\
\Rc^{q}_{21,n} \qquad
\Rqp_n
\end{gathered}\right]\right)\;,
\end{align}
where $\Hqp = (\U_n^\perp)^T\P_\H$, $\Rqp = (\U_n^\perp)^T \Hdag \Rc (\Hdag)^T\U_n^\perp$, and the off-diagonal covariances are $\Rc^{q}_{12,n} = \U_n^T \Hdag \Rc (\Hdag)^T\U_n^\perp$ and $\Rc^{q}_{21,n} = \left(\Rc^{q}_{12,n}\right)^T$.\\
The joint distribution \eqref{datacov} is not used in this manuscript, but is the core of ongoing work to apply different filters to the projected and orthogonal data.

%The covariance matrix $\P_n  \Hdag \Rc(\Hdag)^T \P_n $ is not invertible, but we can invert some parts of it; the likelihood may be written as 
%\begin{align*}
%p( y^p|u) \propto& \exp\left(-\frac{1}{2}\left(y^p-\P_n x\right)^T {(H\P_n )^\dagger} [HH^T R^{-1} HH^T] {(\P_n  H^T)^\dagger}\left(\tilde y^p-\P_n x\right)\right) \;. \\
%\end{align*}
%
%\begin{lem}\label{MPlemma}
%If $H$ is full rank and $R$ is invertible, then 
%\[
%Y^\dagger :=\left(\P_n  \Hdag \Rc(\Hdag)^T \P_n \right)^{\dagger} = (H \P_n )^\dagger [HH^T R^{-1} HH^T] (\P_n  H^T)^\dagger 
%\]
%is a Moore-Penrose pseudo inverse.
%\end{lem}
%
%\begin{proof}
%We need to check the four M-P conditions. Write $X=\P_n  \Hdag L$ and 
%$Y=XX^T$ where $R = L L^T$ is the Cholesky factorization of $R$. 
%First we show that $X$ is a M-P pseudo inverse. Let $Z = \P_n  H^T$ and
%$W = (HH^T)^{-1} L$. The symmetry conditions are
%clear, $X^\dagger X X^{\dagger} = W^{-1}Z^\dagger ZWW^{-1}Z^\dagger = W^{-1}Z^\dagger ZZ^\dagger = W^{-1}Z^\dagger = X^{\dagger}$, and $X X^\dagger X = X$ similarly. Next, we show that $Y=XX^T$ has pseudo inverse $Y^\dagger = (X^\dagger)^T X^\dagger$. Again the symmetry conditions are clear. We have $YY^\dagger Y = XX^T (X^\dagger)^T X^\dagger X X^T = [X (X^\dagger X)^T][(X^\dagger X)X^T]
%= [X (X^\dagger X)][(X^\dagger X)^T X^T] = X X^T = Y$ and similarly for 
%$Y^\dagger Y Y^\dagger = Y^\dagger$.
%\end{proof}

\section{Algorithms for Projected DA} \label{sec:PDA}
In this section we discuss how some combination of the standard/projected forecast models \eqref{model}, \eqref{PIPsys} and data models \eqref{data}, \eqref{pdata2}, \eqref{IPdata}--\eqref{datacov} may be used to form a `projected DA scheme'.  \\
A projected data model changes the innovation, the observation operator, and the observation error covariance. A projected physical model changes the prior and model error covariances.
%Potential for better understanding of interaction between physical and data models in different projected spaces.
We want combinations of physical models, data models, and DA techniques that optimize the assimilation, particularly of the Particle Filtering schemes discussed in Section~\ref{sec:PF}.

% The feature DA framework will be described in full in its' own section. It has not hitherto been treated as a projected DA scheme, and the contribution of this paper is to reformulate the \emph{ad hoc} prior formulations of `feature DA' as explicit schemes that use projected data and/or projected models with a particular choice of the projection $\P_n$.

We identify the following approaches to assimilating with projected data using the results of this paper:
\begin{alg}[Project data only, and discard the orthogonal component] \label{alg:pd}
Apply a standard DA scheme using the unprojected forecast model \eqref{model}, but replace the standard data \eqref{data} with the projected data $y^q_n$ of \eqref{pdata2}. The observation operator is replaced by $\Hq_n$, and the covariance matrix of the observations is replaced by $\Rq_n$.
\end{alg}
A Particle Filter employing Algorithm~\ref{alg:pd} that we denote by PROJ-PF will be tested on a stiff dissipative linear system in Section~\ref{sec:lin}. {PROJ-PF uses the standard forecast model \eqref{model} to update the particles, but computes the weight update with
\begin{align}\label{projW}
w_n^i \propto& \exp\left[-\frac{1}{2} \left(y_n^q - \Hq_n u_n\right)^T \left(\Rq_n\right)^{-1} \left(y_n^q - \Hq_n u_n\right) \right] w_{n-1}^i \;.
\end{align}}

Another algorithm to be described is a novel, efficient PF scheme taking advantage of the Optimal Proposal PF described in section~\ref{sec:op}. 
%The scheme differs from an Algorithm~\ref{alg:pd} implementation of OP-PF by using the full observations to form the proposal, and differs from an Algorithm~\ref{alg:pboth} implementation by not using the orthogonal data $y^{q\perp}_n$ in the assimilation step. 
\begin{alg}[PROJ-OP-PF: Blend projected and unprojected data in the assimilation step] \label{alg:oppfaus}
This algorithm describes a Particle Filter. PROJ-OP-PF uses the typical optimal proposal equations \eqref{opS}--\eqref{opM} for the particle update. The weight update for each particle is computed using the projected data model only, i.e. using {the projected form of \eqref{opW},
\begin{align}
\label{projopW} w_n^i \propto& \exp\left[-\frac{1}{2}(I_n^q)^T\left(\Hq_n\Sc(\Hq_n)^T + \Rq_n \right)^{-1}(I_n^q)\right] w_{n-1}^i \, .
\end{align}}
where $I_n^q \equiv I_n^q(u_{n-1}^i) := y_{n}^q - \Hq_n F_{n-1}(u_{n-1}^i)$.
\end{alg}
Algorithm~\ref{alg:oppfaus} uses all available data to update the particles, but only updates the weights based on how well the particles represent the projected data. This strategy will be tested on the chaotic Lorenz-96 system in Section~\ref{sec:l96}. One major advantage of this approach is that it requires no modification of the numerical simulation used to obtain the forecast. A second advantage is its efficiency; the full data are used for the particle update step, over which the update is straightforward and the dimension of the data does not lead to filter degeneracy; and only the projected data are used to avoid filter degeneracy in the weight update step. The scheme will prove to be more accurate than either, OP-PF or an Algorithm~\ref{alg:pd} implementation of OP-PF, in numerical tests.

We make the following modification to resampling in PROJ-OP-PF:
\begin{alg}[PROJ-RESAMP: Resampling in the Unstable Subspace] \label{alg:r}
When adding noise to particles after resampling, generate (the usual) noise sampled from $\N(\bf{0},\omega^2\I), $ where $\omega\in\R$ must be selected or tuned, and then multiply this random vector by $\alpha\P_n+(1-\alpha)\I$, for some $\alpha\in[0,1].$ 
\end{alg}
When $\alpha=0$ this algorithm is no different to the normal resampling approach, but for $\alpha>0$ some proportion of the uncertainty in resampling is constrained to lie in the space spanned by the columns of $\U_n$. For AUS the resampling scheme should add more noise in the directions of greatest uncertainty in the forecast model, which provides one advantage; a second advantage is that the algorithm does not shift particles as far off the attractor.

\subsection{Convergence results for projected algorithms} \label{sec:conv}
A normal line of inquiry for a new DA algorithm is to quantify the conditions under which it will well represent the posterior distribution, which neglecting time subscripts we write as $p(u|y)$. The projected algorithms above do not generally converge to $p(u|y)$, and so there are two questions: `Does the algorithm converge to a known distribution?', and 'How different is that distribution to the usual posterior?'.

Algorithm~\ref{alg:pd} clearly implements an approximation of the distribution $p(u|y^q)$. 
%and similarly Algorithm~\ref{alg:pdpm} approximates $p(\P_n u | y^q)$. 
That is, a Particle Filter implementation would converge to $p(u|y^q)$ in the limit as the number of particles approaches infinity. 
%The distribution approximated by Algorithm~\ref{alg:pboth} is problem specific and may not be easy to write down, and similarly 
The distribution approximated by Algorithm~\ref{alg:oppfaus} is a blending of $p(u|y)$ and $p(u|y^q)$ that is non-trivial to obtain in closed form.

We now quantify how the Algorithm~\ref{alg:pd} distribution $p(u|y^q)$ relates to the standard posterior $p(u|y)$. For this we will employ the Hellinger distance: given two probability measures $\mu$ and $\mu'$, with associated probability distributions $\rho$ and $\rho'$, the Hellinger distance between the two is
\begin{align} \label{Hell}
d_H (\mu,\mu') = \left[ \frac{1}{2} \int\! \left( \sqrt{{\rho(u)}} - \sqrt{{\rho'(u)}}\right)^2 \, du \right]^{1/2} \;.
\end{align}
The Hellinger distance constrains the difference between functions in the two probability spaces, $|\E^\mu f(u) - \E^{\mu'}f(u)| \le C\, d_H(\mu,\mu')$, true for any $f$ that is square integrable over $\mu$ and $\mu'$ \citep{Law2015}.

To bound this distance for Algorithm~\ref{alg:pd} we write $\rho(u) = p(u|y)$ and $\rho'(u) = p(u|y^q)$. The second distribution is written as 
$$p(u|y^q) = p(u|y) \frac{ p(y|y^q)}{p(y|u,y^q)} \;,$$
obtained via Bayes' law in the form $p(u) = p(u|y)\, p(y) / p(y|u)$, conditioning on $y^q$, and using $p(u|y,y^q) = p(u|y)$. Using $p(y|y^q) = p(y^{q\perp}|y^q)$, we obtain the final form
$$p(u|y^q) = p(u|y) \frac{ p(y^{q\perp}|y^q)}{p(y^{q\perp}|u,y^q)} \;.$$

Substituting into \eqref{Hell} we obtain a bound for the consistency of Algorithm~\ref{alg:pd} with the original posterior $p(u|y)$,
\begin{align}
\nonumber d_H(\mu,\mu') &= \left[ \frac{1}{2} \int\! \left( 1 - \sqrt{\frac{ p(y^{q\perp}|y^q)}{p(y^{q\perp}|u,y^q)}}\right)^2 \rho(u) \, du \right]^{1/2} \\
\label{refHell} &= \left[ \frac{1}{2} \E^\mu \left( 1 - \sqrt{\frac{ p(y^{q\perp}|y^q)}{p(y^{q\perp}|u,y^q)}}\right)^2 \right]^{1/2}\;.
\end{align}
% \note{Can rewrite bound in terms of $$\frac{p(y^q|u)}{p(y|u)}p(y^{q\perp}|y^q).$$
% If the projected data constrains hte orthogonal data (slow/fast variables) then $p(y^{q\perp}|y^q)\approx 1$; we know the explicit form of $p(y^q|u)/p(y|u)$. Check again before including.}
Intuition on the projected algorithms suggests that if the projection somehow represents `important' quantities in the model, e.g. directions associated with positive Lyapunov exponents, or coherent structures, etc., then the projected data will retain the same key information from the original data, and the posterior approximated by the projected DA algorithm will be similar to the original posterior. The above result quantifies that intuition. The posterior distribution $p(u|y^q)$ associated with the projected algorithm will be close to $p(u|y)$ provided that knowing the projected data $y^q$ is about as useful as knowing the truth $u$ in determining the values of the orthogonal, discarded data; in that case $p(y^{q\perp}|y^q) \approx p(y^{q\perp}|u,y^q)$ and $d_H(\mu,\mu')\approx 0$. %The ratio can be written as $p(y^{q\perp}|y^q)/p(y^{q\perp}|u,y^q)$ to further clarify that knowledge of the projected data should constrain the values of the orthogonal projected data.

%A similar bound may easily be established for Algorithm~\ref{alg:pdpm}. Let now $\rho'(u) = p(\P_n u| y^q)$, and $\mu'$ the associated probability measure; then
%\begin{align*}
%d_H(\mu,\mu') %&= \left[ \frac{1}{2} \int\! \left( 1 - \sqrt{\frac{p(\P_n u|y)}{p(u|y)}\frac{ p(y|y^q)}{p(y|\P_n u,y^q)}}\right)^2 \rho(u) \, du \right]^{1/2} \\
%&= \left[ \frac{1}{2} \E^\mu \left( 1 - \sqrt{\frac{p(\P_n u|y)}{p(u|y)}\frac{ p(y|y^q)}{p(y|\P_n u,y^q)}}\right)^2 \right]^{1/2}\;.   
%\end{align*}
%The modification to this expression from the previous one is a second ratio that measures the information gap between the distribution of $u$ and of $\P_n u$. As with the previous expression, the ratio is far from $1$ if knowledge of $\P_n u$ does not constrain the possible values of $(\I-\P_n)u$ at all, and close to $1$ if the projected data does have some relation to the full state. 

An intuitive example of the above bounds in practice is a slow-fast system with a slow manifold onto which the fast variables are attracted. Choosing $\P_n$ to identify the slow variables will lead to a small value of $d_H(\mu,\mu')$ for either Algorithm~\ref{alg:pd},
%or \ref{alg:pdpm}, 
since knowledge of the slow variables is sufficient to constrain the fast variables. In the case where there are few slow variables and many fast variables, then, an Algorithm~\ref{alg:pd} Particle Filter will be a much less degenerate implementation of the Particle Filter that converges close to the desired posterior $p(u|y)$. A linear system of this type will be the first numerical example, in Section~\ref{sec:lin}.

%%%%%%%%%%%%%%%%%%%%%%%%%%%%%%%%%%%%%%%%%%%%%%%%%%%%%%%%%%%%%%%%%%%%%%%
%\section{Applications} \label{sec:Applications}
\section{Application: Assimilation in the Unstable Subspace} \label{sec:AUS}

% There is an increasing understanding of the
% importance of the interaction between the model dynamics and its relationship to the observations operator(s) employed, see \cite{TrDiTa10,PaCaTr13, ToHu13, Law16,Bocquet17}.  
% Historically AUS has tried to ensure the forecast projects strongly into the unstable/neutral subspace (ETKF, 4D-Var, etc). Here the approach is different. We assume that the forecast has some projection into the US, and instead focus on the data. This approach is suitable for many modern DA methods that use ensembles, particularly the EnKF and PF. 

For the remainder of the paper we will study the case where the projection identifies the most unstable modes in the forecast model. To determine these modes we employ the discrete QR algorithm \citep{DVV07,DiVV15}.
For the discrete time model $u_{n+1} = F_n(u_n) + \sigma_n$ with $u_n\in\R^N$, let $\U_0\in\R^{N\times p}$ $(p\leq N)$ denote a random matrix such that $\U_0^T \U_0 = \I$,
\begin{align}\label{QR}
\U_{n+1} \T_n =& F_n'(u_n)\U_n \approx \frac{1}{\epsilon}[F_n(u_n + \epsilon \U_n) - F_n(u_n)],\,\,\, n=0,1,...
\end{align}
where $\U_{n+1}^T \U_{n+1} = \I$ and $\T_n$ is upper triangular with positive diagonal elements.
With a finite difference approximation the cost is that of an ensemble of size $p$ plus a reduced $QR$ via modified Gram-Schmidt to re-orthogonalize. Time dependent orthogonal projections to decompose state space are $\P_n = \U_n \U_n^T$ and $ \I-\P_n = \I - \U_n\U_n^T$. 
\jm{In order to apply \eqref{QR} to an ensemble DA method, $\U_0$ must be specified and we must choose how to obtain $u_n$ from the ensemble of particles at time $t_n$. We initialise $\U_0$ from a modified Gram-Schmidt orthonormalization of a random $N\times p$ matrix, and choose $u_n := \sum_i w_n^i u_n^i$, the weighted particle mean.}
%where as described in Section~\ref{sec:usProj} $\P_n$ refers either to $\P_n$ or to the intersection of $\P_n$ and $\P_\H$.

\subsection{A comparison of the projected approach to classical AUS techniques}
\label{sec:comp}
This somewhat technical section establishes the relationship between existing AUS algorithms and the projected data approach. We consider the EKF-AUS \citep[e.g.]{Trevisan11,PaCaTr13}. EKF-AUS is a modified EKF in which the forecast covariance matrix $\Pf_n$ is replaced by the projected matrix $\P_n \Pf_n \P_n$, leading to the Kalman gain
\begin{align}
\label{ausk}
    \Kn =& \P_n \Pf_n \P_n \H^T \left[ \H \P_n \Pf_n \P_n \H^T + \Rc\right]^{-1}\;,
\end{align}
where the EKF forecast covariance matrix $\Pf_n$ and observation operator $\H_n\equiv\H$ are described in Section~\ref{sec:EKF}. It is clear that the EKF-AUS Kalman gain can be written as a combination of the columns of $\U_n$. 

For comparison, we write down the Kalman gain associated with the data model \eqref{pdata1},
\begin{align}
\label{prok}
    \Kn = \Pf_n \P_\H \P_n \left[\P_n\Hdag\left(\H \Pf_n \H^T + \Rc\right)(\Hdag)^T\P_n\right]^{\dag} \;.
\end{align}
We choose this form to most closely resemble EKF-AUS; the arguments of Theorem~\ref{thm:main} guarantee that \eqref{prok} is identical to the Algorithm~\ref{alg:pd} implementation of the EKF.\\
 The difference between the two Kalman gains is essentially that \eqref{prok} interchanges the position of $\H$ and $\P_n$, requiring the use of $\Hdag$ in order to do so, but manages to project all terms in the covariance-weighting inverse instead of only the forecast covariance matrix. Unlike the classical AUS gain \eqref{ausk}, \eqref{prok} does not restrict the analysis increment to the unstable subspace. The innovation is $y_n - \H u^f_n$ in classical AUS, but with \eqref{prok} would be $y^p_n - \P_n \P_\H u^f_n$.\\
That is, classical AUS uses the full data but restricts the assimilation update to the unstable subspace via \eqref{ausk}; Algorithm~\ref{alg:pd} restricts the innovation to the unstable subspace but the assimilation update can distribute this innovation across the whole of model space. The comparison between these algorithms here is pedagogical, not competitive; the advantages of the EKF-AUS algorithm are well established, while Algorithm~\ref{alg:pd} effects a reduction in data dimension that we will explore for Particle Filters, not the EKF.

Finally we obtain a form of EKF associated with the projected model \eqref{PIPsys} and unprojected data. This is essentially a re-derivation of EKF-AUS from the projected framework employed in this paper, confirming that the two are compatible. Consider the linearized physical model $u_{n+1} = \An u_n + \sigma_n$ of Section~\ref{sec:EKF}. 
Then the projected physical model has the form $\P_{n+1}u_{n+1} = \P_{n+1} \An \P_n u_n + \P_{n+1}\sigma_n$ or
\begin{align*}
v_{n+1} =& \P_{n+1} u_{n+1} = [\U_{n+1} \T_{n} \U_n^T] v_n +  \P_{n+1}\sigma_n 
\\
\equiv& \B_n v_n + \P_{n+1}\sigma_n,
\end{align*}
where $v_n = \P_n u_n$, $\B_n = \An \P_n$, and using $\P_{n+1}\An\P_n = \U_{n+1}\T_n\U_n^T.$
% The Kalman gain for the original, unprojected physical and data models is given by
% \[
% K_{n+1} \equiv K(\hat C_{n+1}) = \hat C_{n+1} \H^T (\H \hat C_{n+1} \H^T + \Rc)^{-1},
% \]
% where the forecast covariance matrix $\hat C_{n+1} = \An C_n \An^T + \Sc$, and where $C_n$ is the analysis covariance matrix.
% For the projected physical model and the unprojected data model we have 
The forecast covariance matrix is 
\begin{align*}
    \tPf_{n+1} =&  \P_{n+1} \An  \tPa_{n} \An^T  \P_{n+1} +  \P_{n+1}\Sc  \P_{n+1}
    \\
    =& \P_n \tPf_{n+1} \P_n
    %=&\B_n (\P_n \tPa_n \P_n) \B_n^T + \P_{n+1}\Sc\P_{n+1}
\end{align*}
Initialising $\tPa_0 = \Pa_0$, then $\tPf_n$ is precisely the EKF-AUS forecast covariance matrix. 

We will now explore the benefits of the projected data algorithms in an AUS framework, using \eqref{QR} to calculate the projections. %For the remainder of the paper we refer to the projected DA schemes as AUS schemes, e.g. PROJ-PF for an Algorithm~\ref{alg:pd} Particle Filter and PROJ-OP-PF for an Algorithm~\ref{alg:oppfaus} Optimal Proposal Particle Filter.
The first test case is a simple linear model that demonstrates the benefit of reducing the data dimension using PROJ-PF.

\subsection{Case study: linear model with Gaussian noise} \label{sec:lin}
Suppose that forecasts are made for $u\in\R^{100}$ with the model %a physical model of the form
\begin{align}\label{projlin}
%\frac{du}{dt} = \mathbf{A} u + \sigma \dot W \;,
u_n = e^{\mathbf{A}(t_n-t_{n-1})}\, u_{n-1} + \sigma_n
\end{align}
where %$W$ is a Wiener process.
$\sigma_n\sim\N(0,\,0.05\,\I_{100})$. %The one-step model \eqref{model} is obtained through a simulation of this physical model with a standard numerical integrator, initialised at one observation time and terminating at the next.
We construct $\mathbf{A}\in\R^{100\times 100}$ so that it has two eigenvalues with small real part $Re(\lambda_i) \in (0,0.04)$, and so that the remaining 98 eigenvalues have real part $Re(\lambda_i) \le -100$. This produces a well known multiscale dynamic, which we describe for the underlying deterministic physical model $du/dt = \mathbf{A}u$. There exists a transformation of this system into a system consisting of  2 `slow' and 98 `fast' variables. The fast variables are rapidly attracted onto a slow invariant manifold that depends only on the slow variables. After an initial transient, the system is effectively 2-dimensional. 
%Using centre manifold theory \cite{Carr} one can obtain an expression for this slow manifold, and furthermore calculate an ordinary differential equation that governs the evolution of the 2 slow variables. \\
We run DA experiments {assuming that the slow manifold and reduced system are unknown,} instead 
using forecasts and observations from the full, 100-dimensional system \eqref{projlin}. %A successful assimilation step will be to get the correct distribution for the (unknown) 2-dimensional subspace of model space corresponding to the slow variables. 

%We will compare results from the Particle Filter to PF-AUS. We have described the poor performance of the Particle Filter in high dimensions; this system, despite the low-dimensional structure of the underlying dynamics, is no different.  The AUS methodology is designed for more realistic scenarios than this model, in particular for scenarios in which the low-dimensional subspace is time-varying; but AUS will readily approximate the reduced subspace in this simpler problem.\\

We present results for the PF compared to PROJ-PF using Algorithm~\ref{alg:pd}. Four scenarios are considered: where every variable is observed, every second variable, every fourth variable, and finally a scenario in which only the first and 51st variables are observed. \\

Let us pause here to predict the results. The PF weight update depends crucially on the statistical distance of the observations from each particle, the exponent of \eqref{likelihood}. As the dimension of the data increases, the statistical distance of each particle and each observation from the attractor increases due to the accumulation of terms from the measurement error and model noise. The key information - about the distance of each particle from the observation in the 2-dimensional slow subspace that governs the dynamics - is swamped by the accumulation of errors in the less significant 98-dimensional fast subspace. We expect the PF to perform well when the data is 2-dimensional, but grow steadily worse as the data dimension increases. The algorithm of PROJ-PF, by contrast, will estimate the low-dimensional subspace in which the dynamics occurs and confine the data and (through the observation operator) the forecast to this subspace when performing the assimilation. By doing so the dimension of the model and data should affect the accuracy of the algorithm much less.

The remaining experiment parameters are as follows. Particles are initialised at time $t_0$ from a Gaussian with a standard deviation of $0.2$ and initial bias of $0.22$ from the randomly drawn true initial condition. We set $t_n = 0.1n$, and simulate the truth using \eqref{projlin}, collecting observations every $0.1$ time units with small measurement error covariance $\Rc = 0.05^2\I$, until 100 observation times have passed. Both PF algorithms resample if the ESS drops below half the number of particles, which is 1000. On resampling, noise is added to every variable with a standard deviation of $0.02$. We use PROJ-PF with $p=2$. \\

We report the Root Mean-Squared Error (RMSE) between the filter mean at each time step and the true system state. The standard Particle Filter performs very poorly with high dimensional data, while PROJ-PF is reasonably indifferent to the dimension of the data and in all cases has mean RMSE below the RMSE of the observations. Results are displayed in Figure~\ref{fig:lin}. \\
The two extremes of the data dimension serve to highlight its role in Particle Filter divergence, and the role of PROJ-PF. In Figure~\ref{fig:lin0} every variable is accurately observed, and consequently one could obtain a reasonable estimate of the system at every observation time by discarding the model and using the data. Despite this, and despite the low-dimensional attractor in the state dynamics, the Particle Filter estimate diverges frequently and far from the true state. The other extreme in data availability is Figure~\ref{fig:lin3}, in which only two variables are observed and the Particle Filter has an accurate mean RMSE of 0.03. By comparison PROJ-PF is more accurate than the observations in each scenario, and in particular does not diverge at large data dimension.

\begin{figure}
\centering
\subfloat[][All 100 variables are observed. Mean RMSE: 0.21 PF, 0.04 PROJ-PF.]{
\includegraphics[scale=0.46]{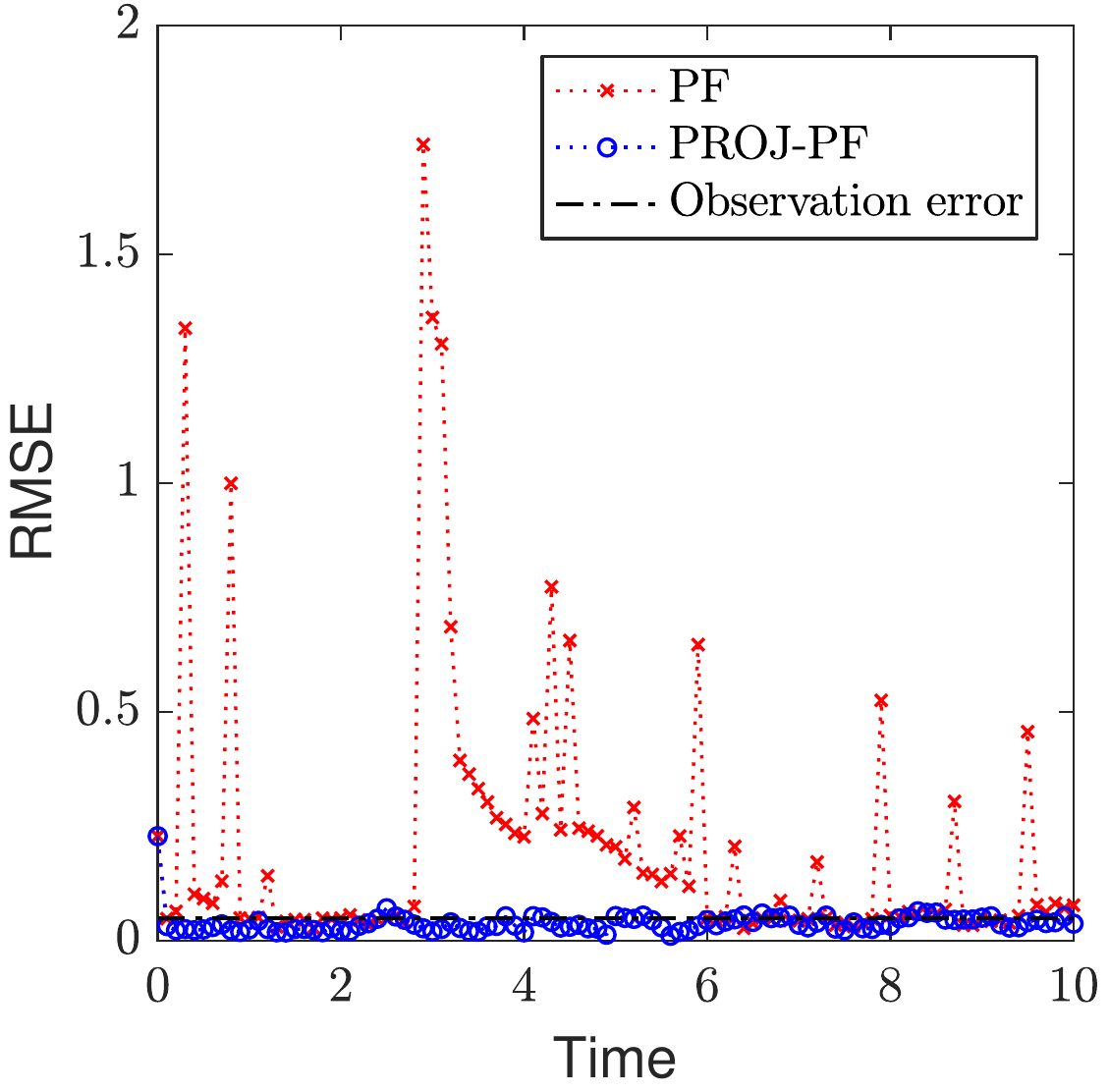}
\label{fig:lin0}}
\quad
\subfloat[][Observations of 50 evenly spaced variables. Mean RMSE: 0.09 PF, 0.04 PROJ-PF.]{
\includegraphics[scale=0.45]{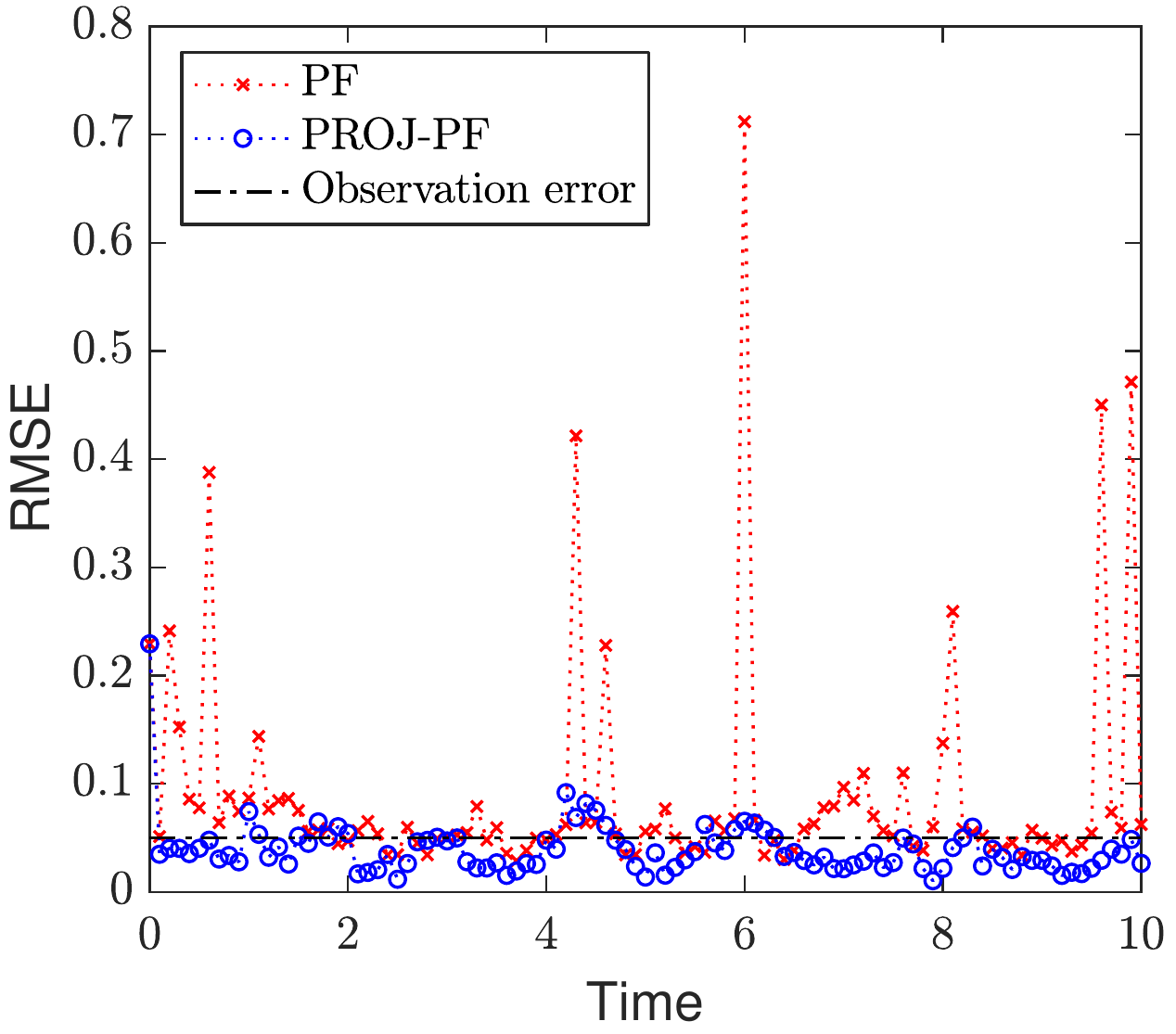}
\label{fig:lin1}}

\subfloat[][Observations of 25 evenly spaced variables. Mean RMSE: 0.15 PF, 0.05 PROJ-PF.]{
\includegraphics[scale=0.45]{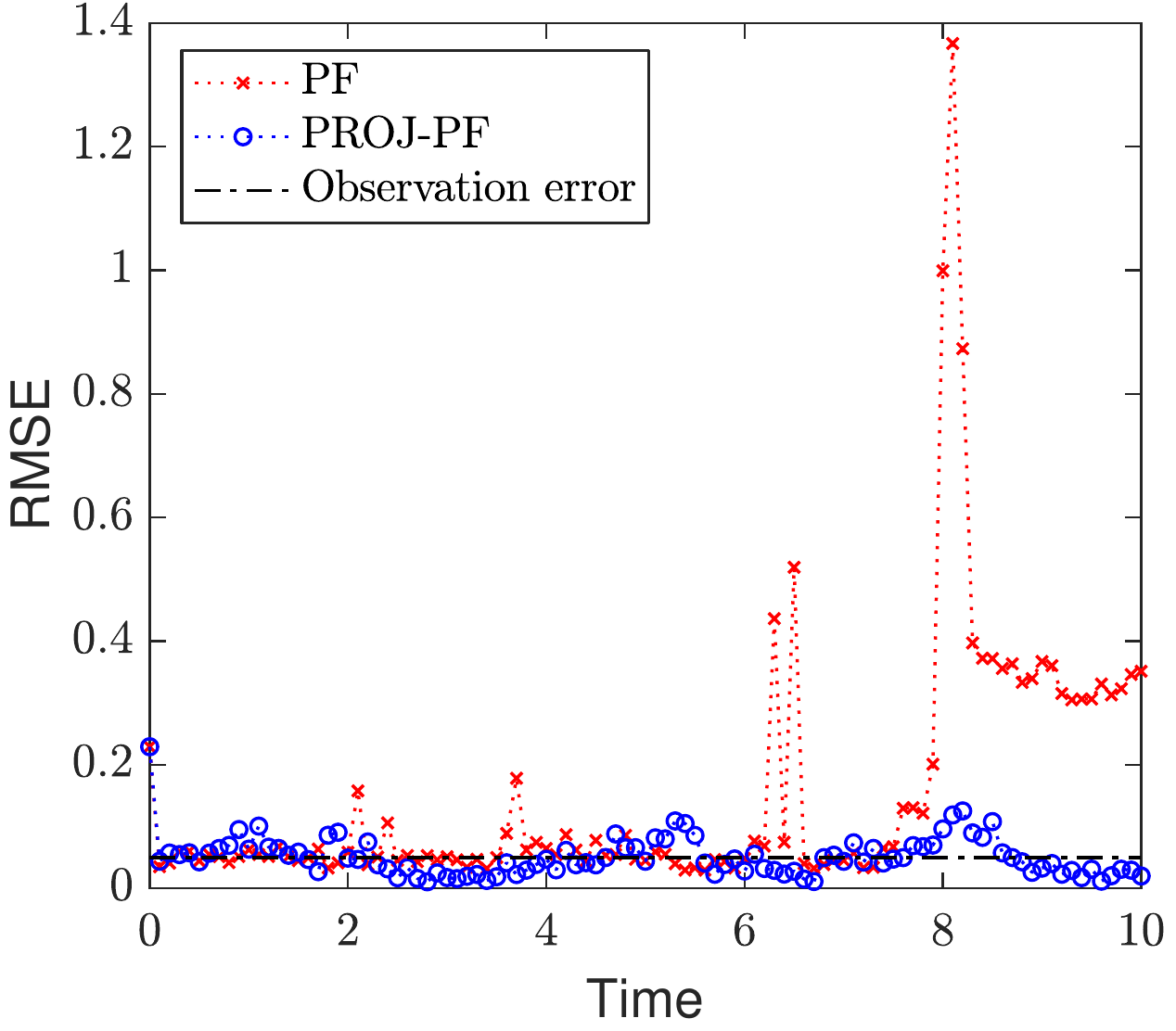}
\label{fig:lin2}}
\quad
\subfloat[][Observations of the first, and 51st, components of $u$. Mean RMSE: 0.03 PF, 0.03 PROJ-PF.]{
\includegraphics[scale=0.45]{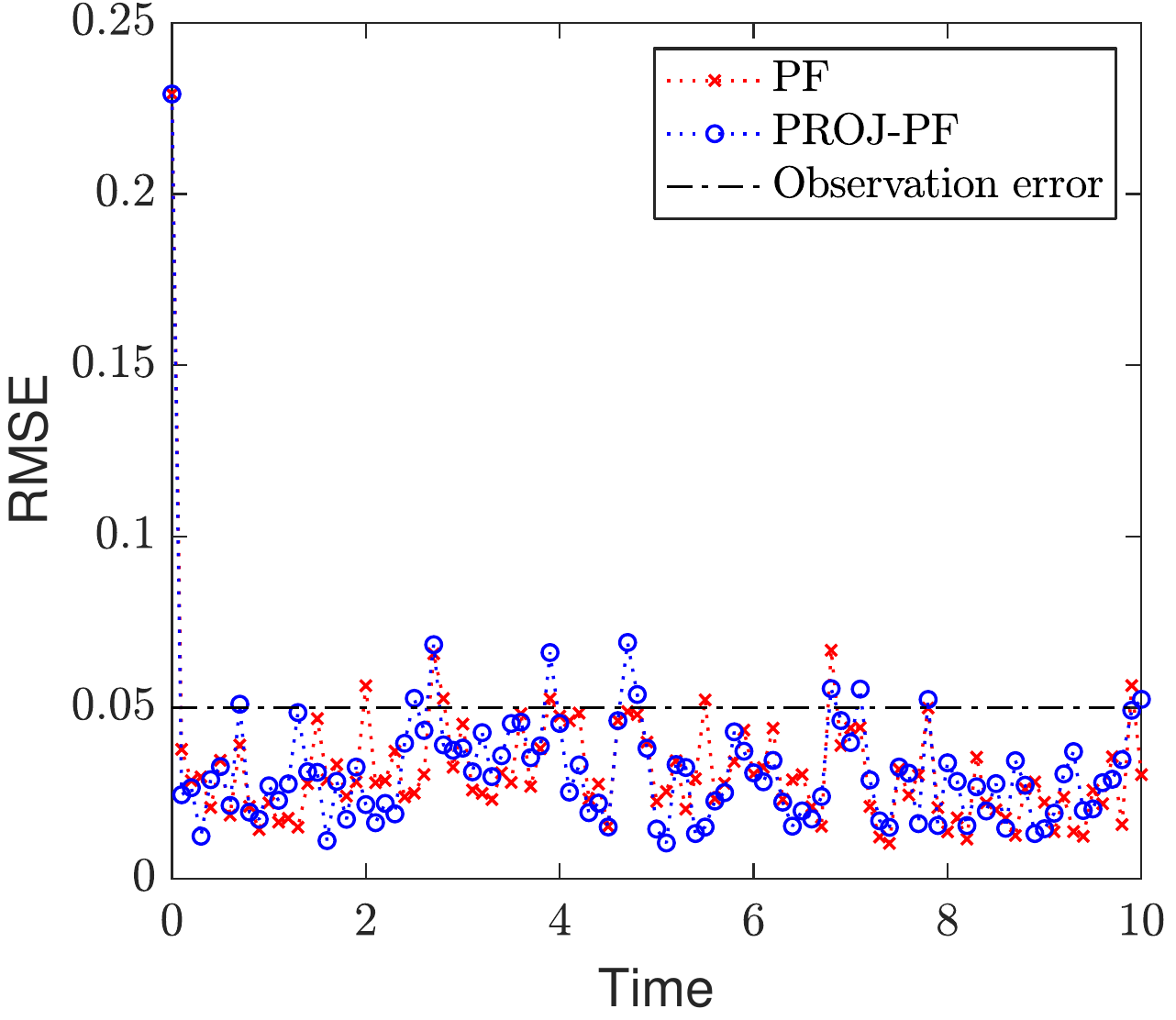}
\label{fig:lin3}}
\caption{Comparison of the Particle Filter to an Algorithm~\ref{alg:pd} implementation of PROJ-PF for a linear system as the number of variables observed is changed. The PF diverges with increasing data dimension, but can accurately capture the posterior with observations of any two random variables.}
\label{fig:lin}
\end{figure}

{On longer time intervals the PF RMSE increases significantly in the cases where the data is 25-, 50-, and 100-dimensional. The PROJ-PF algorithm remains stable and accurate in all scenarios.}

We now present examples from the Lorenz 96 system.%, in which reducing the number of observations will no longer be sufficient for the Particle Filter to perform well. %The key difference will be that the unstable subspace alone is no longer sufficient to well approximate the system dynamics. Our approach will be to develop a scheme that assimilates in both, the P- and I-P- spaces, based on the discussion in Section~\ref{sec:PDA}. The assimilation in the unstable subspace will be performed with a Particle Filter, and as shown above, reducing the dimension of the data to be assimilated will be key to good performance. We will also develop a novel approach to resampling, in which noise is added more strongly in the directions corresponding to the unstable/neutral subspace, and weakly in the stable directions. As we shall see, this approach leads to better error statistics and much reduced degeneracy and corresponding fewer resampling steps, in the PF-AUS algorithm. 

\subsection{Case Study: Chaotic Lorenz 96 system} \label{sec:l96}
Consider the system of ordinary differential equations introduced in \cite{Lo96},
\begin{align} \label{l96}
    \dot u_i =& \left(u_{i+1}-u_{i-2}\right)u_{i-1} -u_i +F\;,
\end{align}
for $i=1,...,J,$  and $F=8$. If $J=40$, then this system is chaotic with $14$ positive and $1$ neutral Lyapunov exponents.
We present experiments in which the deterministic part of the model \eqref{model} is given by an integration of \eqref{l96} for a fixed time. {The true system state is generated by the same procedure; only the initial condition and realizations of the model noise are different.}

The primary focus of this section is Algorithm~\ref{alg:oppfaus}, PROJ-OP-PF, employing PROJ-RESAMP as in Algorithm~\ref{alg:r} when resampling, compared to the OP-PF and EnKF. An Algorithm~\ref{alg:pd} implementation of the ETKF is also considered. 

{In all simulations, observations of every second variable are available, evenly spaced, at each observation time. Observations will generally be accurate (with standard deviation equal to or less than $0.1$), which exacerbates the problem of filter degeneracy that the projected algorithms are intended to mitigate. We will confine experiments to $L=2000$ or $L=50$ particles, the latter of which resembles the affordable ensemble size for geophysical applications. Model simulations are carried out by bridging the observation time step with 5 steps of the fourth order Runge-Kutta scheme.}

We will first consider a regime in which observations are assimilated frequently in time, so each forecast ensemble is strongly contained in the low-dimensional subspace $\P_n$. We then consider longer times between observations, and finally a high-dimensional filtering scenario. In all cases PROJ-OP-PF will significantly outperform the Optimal Proposal PF. The key parameters to be tuned are the projected data dimension $p$, noise added on resampling $\omega$, and confinement to $\P_n$ of the resampling noise, $\alpha$. The latter two parameters were introduced in Algorithm~\ref{alg:r}. OP-PF will be tuned by varying the resampling noise $\omega$. 

\subsubsection{Frequent, accurate observations with a moderate ensemble}
{Set model noise $\Sc = 0.01\I_N$ and dimension $J=40$, number of particles $L=2000$, observation noise $\Rc = 0.01\I_M$, and time between observations to $0.005$ time units. Translating the observation step into dimensional units, this corresponds to assimilating observations every 35 minutes. When an experiment records a time-averaged RMSE, a spinup of 100 assimilation steps is computed and discarded, then error statistics are measured for another 100 steps. Figures~\ref{fig:l96_tune_r}~to~\ref{fig:enkf} are computed in this parameter regime.}

{We first demonstrate how OP-PF and PROJ-OP-PF are tuned. Both algorithms are run with $20$ different values of $\omega$ between $10^{-5}$ and $10^{-2}$, and $10$ values of $\alpha$ between $0$ and $1$. The second parameter $\alpha$ is used only in PROJ-OP-PF, in the PROJ-RESAMP Algorithm~\ref{alg:r}. The mean RMSE and percentage of resampling steps (after the spinup) are recorded, and each algorithm is repeated $20$ times in each configuration. The rank of the projection was chosen to be $p=3$ for PROJ-OP-PF. Figure~\ref{fig:l96_tune_r} shows a sample of the result for PROJ-OP-PF. The optimal choice of $(\omega,\,\alpha)$ is taken to be the choice that minimises the RMSE. Note that the RMSE and filter degeneracy both strictly decrease as $\alpha$ increases, at all considered values of $\omega$. All following figures are produced using an optimal choice of $(\omega,\,\alpha)$ for PROJ-OP-PF, and of $\omega$ for OP-PF. }

\begin{figure}
    \centering
    \includegraphics[width=0.16\textwidth]{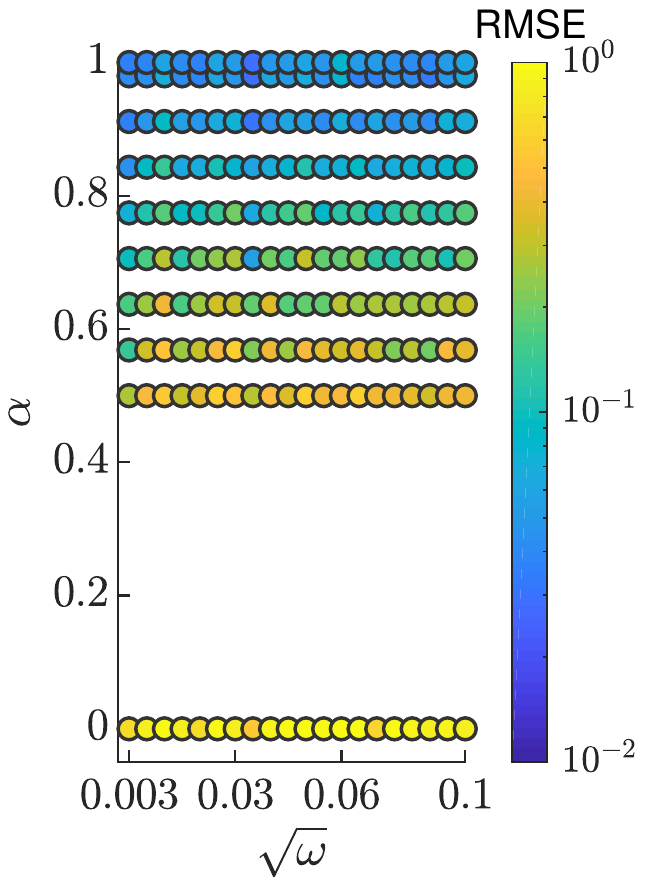} \includegraphics[width=0.16\textwidth]{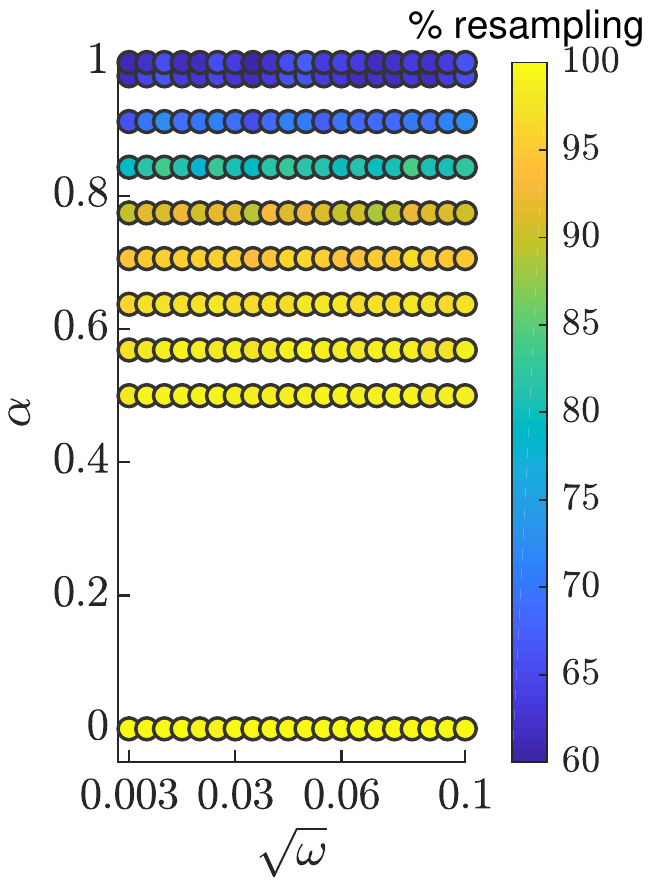}
    \includegraphics[width=0.16\textwidth]{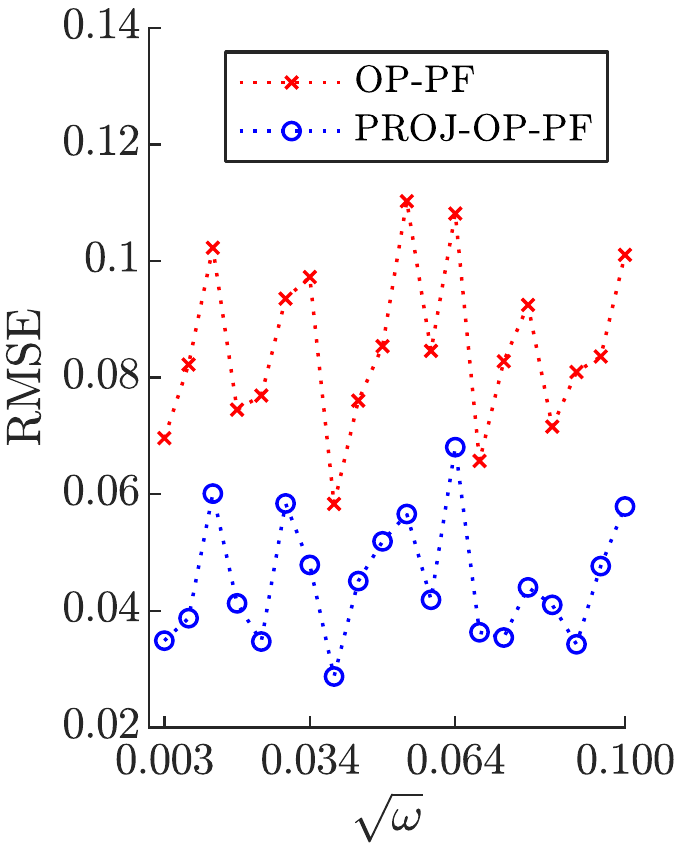}
\caption{\emph{Left, Middle}: Statistics for PROJ-OP-PF with $p=3$ as the resampling noise $\omega$ and confinement to the unstable subspace $\alpha$ are varied for the Lorenz 96 system with time $0.005$ between observations. Data points are shaded to reflect their value. Each data point represents the mean from $20$ repetitions, each of which was also time-averaged. The RMSE, ranging from $0.03$ to $1.3$, decreases with increasing $\alpha$. \emph{Right}: RMSE for OP-PF as $\omega$ varies, compared to PROJ-OP-PF results with the optimal choice of $\alpha$. The mean error for PROJ-OP-PF is $53\%$ of the mean error for OP-PF. }
\label{fig:l96_tune_r}
\end{figure}

%The remaining experiments use Algorithm~\ref{alg:r} for resampling in the AUS algorithms with $\alpha = 0.99$. The resampling noise $\omega$ is optimised separately in the standard PF and the PROJ-PF algorithms by choosing the minimum RMSE result over $20$ values in $[10^{-4},10^{-2}]$.

We now investigate the optimal choice of dimension $p$ for the projected data in PROJ-OP-PF. One might expect that $p\ge15$ would be optimal, as the system has $15$ unstable and neutral modes. However, the blending of projected and unprojected data in Algorithm~\ref{alg:oppfaus} will sufficiently constrain the weakly unstable modes in the system, and at the same time the PF algorithm will avoid degeneracy at low values of $p$. The RMSE and number of resampling steps taken by PROJ-OP-PF compared to OP-PF are displayed in Figure~\ref{fig:l96_tune_p}. The RMSE has a clear minimum at $p=6$, about $30\%$ less than the OP-PF RMSE, and the frequency of resampling in PROJ-OP-PF decreases sharply with $p$. The optimal choice of noise to add on resampling was $\omega = 0.0027$ for OP-PF, and $\omega = 0.056$ for PROJ-OP-PF. The optimal noise for PROJ-OP-PF is an order of magnitude larger than for OP-PF. This suggests that one benefit of the novel resampling scheme is the ability to more vigorously explore the uncertain directions in the forecast without moving system estimates too far off any local attractor.

\begin{figure}
\centering
\subfloat{
\includegraphics[width=0.47\textwidth]{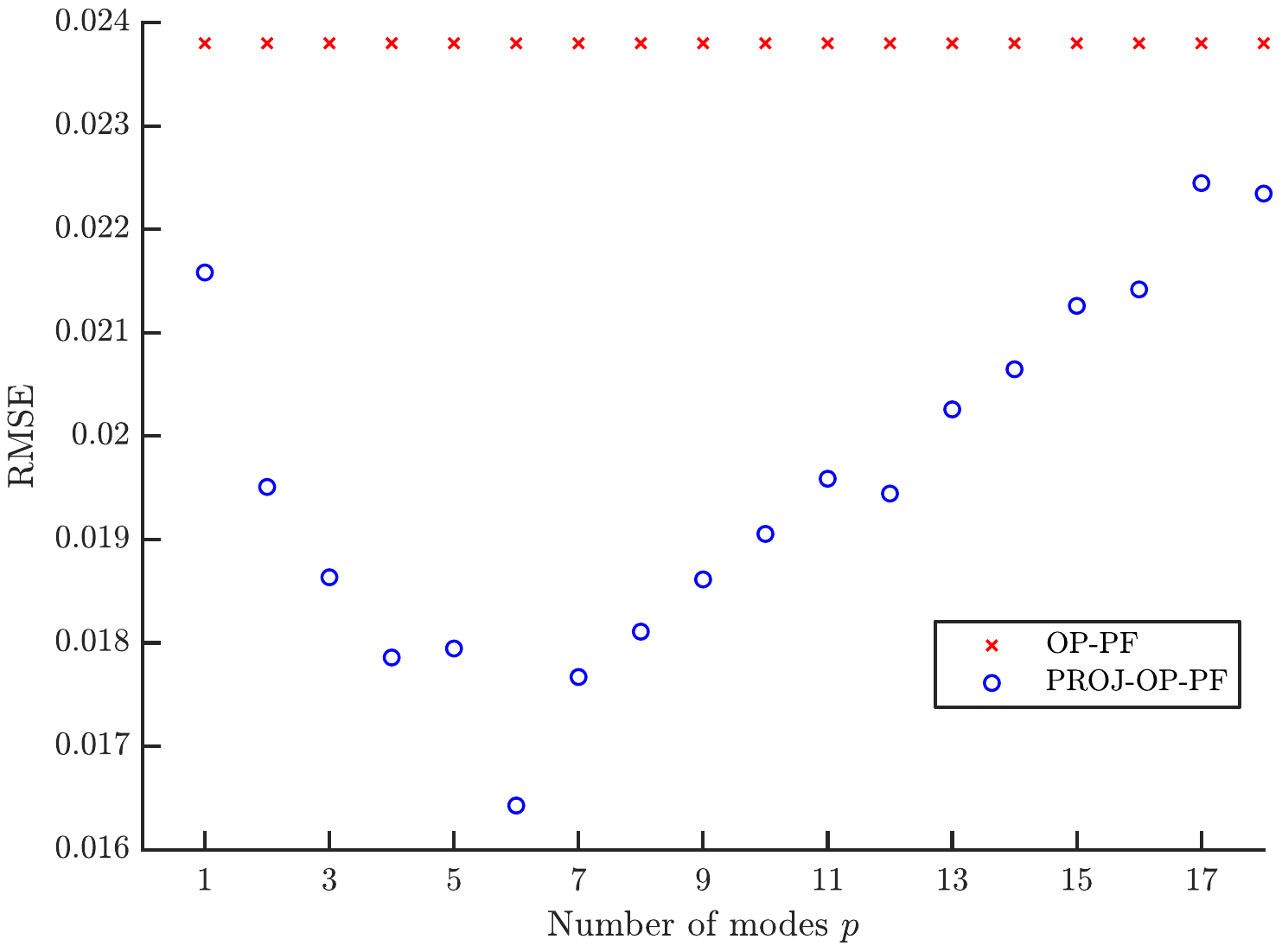}
\label{fig:prmse}}
\quad
\subfloat{
\includegraphics[width=0.47\textwidth]{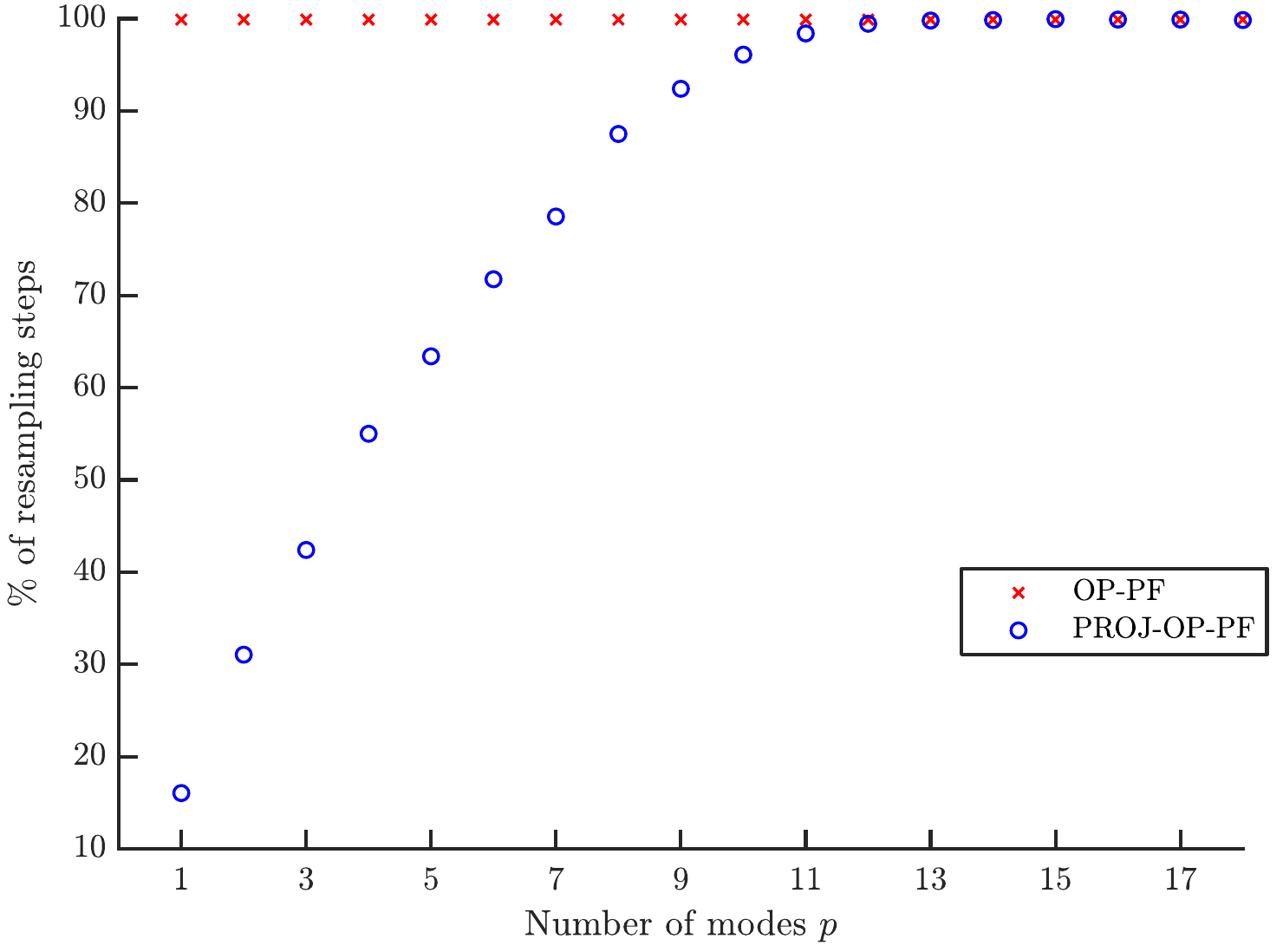}
\label{fig:pdegen}}
\caption{Error statistics for PROJ-OP-PF as the rank of the projection is varied, compared to the Optimal Proposal PF, for the Lorenz96 system with $0.005$ time units between observations. Each data point represents the mean from $20$ repetitions, each of which was also time-averaged.}
\label{fig:l96_tune_p}
\end{figure}

The selection of $p$ in a non-degenerate DA scheme is less crucial. For comparison to Figure~\ref{fig:l96_tune_p} we implement an ETKF and compare to PROJ-ETKF (implemented via the projected data approach of Algorithm~\ref{alg:pd}). The same experimental parameters are used as for the PF results, except for the ensemble size which is 50. Results are shown in Figure~\ref{fig:enkf}. For the PROJ-ETKF the error statistics are similar for a large range $5\le p \le 13$, about $20\%$ below the mean ETKF behaviour\footnote{It is a little surprising that the PROJ-ETKF does any better than ETKF at all, as the ETKF does not suffer from the curse of dimensionality. %An Algorithm~\ref{alg:pd} implementation of the Kalman Filter does not improve on the standard KF, so it is a detail of implementation and not of fundamental structure that PROJ-ETKF has lower mean RMSE than ETKF. 
It may be that the dimension reduction involved in PROJ-ETKF ameliorates ill-conditioning in the the Kalman gain in such a manner that the ETKF is thereby improved; but if so the mechanism of improvement is still not clear, since the ETKF was designed for exactly that scenario already \citep{Bishop2001}. PROJ-ETKF has no benefit to RMSE in the more realistic scenario where observations are assimilated less frequently in time.}.  %It seems likely that the optimal choice of $p$ for the Particle Filtering AUS methods in Figure~\ref{fig:l96_tune_p} is a tradeoff between obtaining more informative data by increasing p, and avoiding filter degeneracy by reducing p.

\begin{figure}
    \centering
    \includegraphics[width=0.47\textwidth]{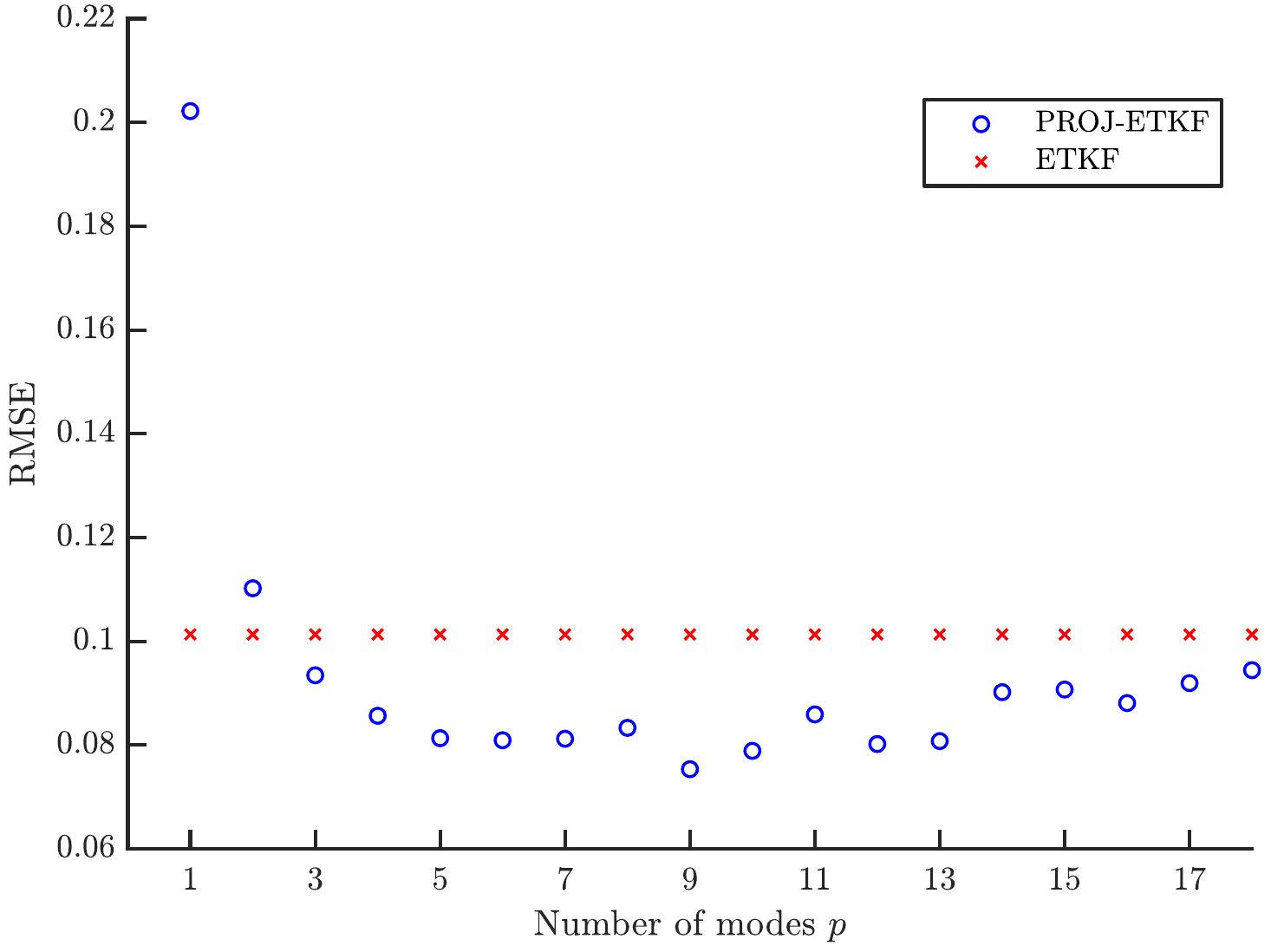}
    \caption{Statistics for PROJ-ETKF as the rank of the projection is varied, compared to the Ensemble Transform Kalman Filter, for the Lorenz96 system. Each data point represents the mean from $20$ repetitions, each of which was also time-averaged. Somewhat surprisingly, projecting the data reduces the error in the EnKF for this experiment, though to a lesser extent than for the PF methods.}
    \label{fig:enkf}
\end{figure}

\subsubsection{Infrequent, accurate observations with a small ensemble}

{We now move to a more realistic scenario in which observations are infrequent and the affordable ensemble size is small. We preserve model noise $\Sc = 0.01\I_N$ and dimension $J=40$, and observation noise $\Rc = 0.01\I_M$, but set the number of particles $L=50$ and the time between observations to $0.05$ time units. Translating the observation step into dimensional units, this corresponds to assimilating observations every 6 hours. When an experiment records a time-averaged RMSE, a spinup of 200 assimilation steps is computed and discarded, then error statistics are measured for another 100 steps. Figures~\ref{fig:l96_bOP}~and~\ref{fig:real} are computed in this parameter regime.}

As in the previous section, the main result is to show how scaling the projected data dimension $p$ affects the RMSE, and in particular when, or if, PROJ-OP-PF outperforms the OP-PF. This scaling is shown in Figure~\ref{fig:l96_bOP}, in which the best PROJ-OP-PF results achieves mean RMSE 2/3 of that of the OP-PF. The percentage of steps that trigger resampling is also shown. As before, it monotonically increases with $p$.

\begin{figure}
\centering
\subfloat{
\includegraphics[width=0.47\textwidth]{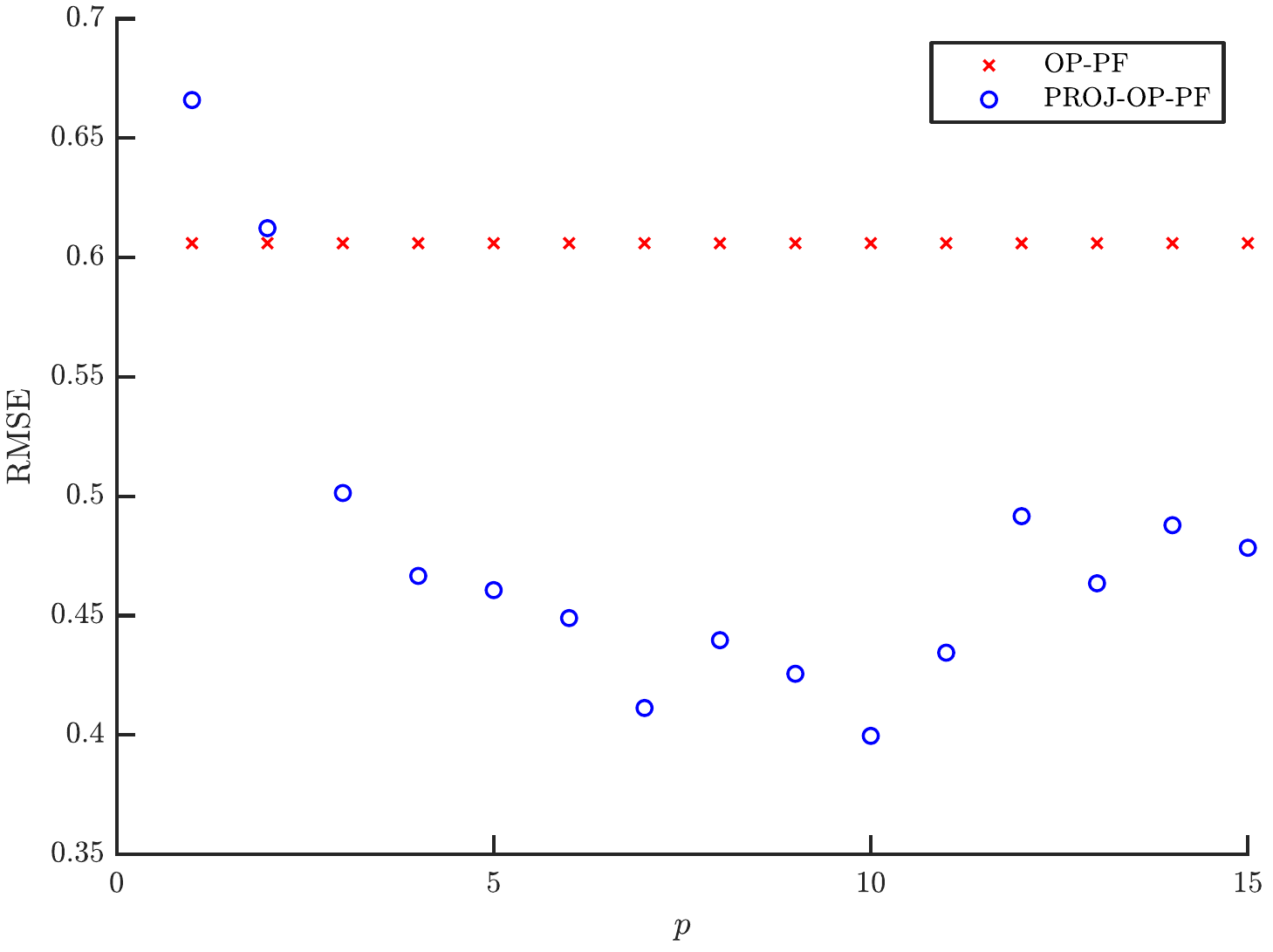}
\label{fig:brmse}}
\\
\subfloat{
\includegraphics[width=0.47\textwidth]{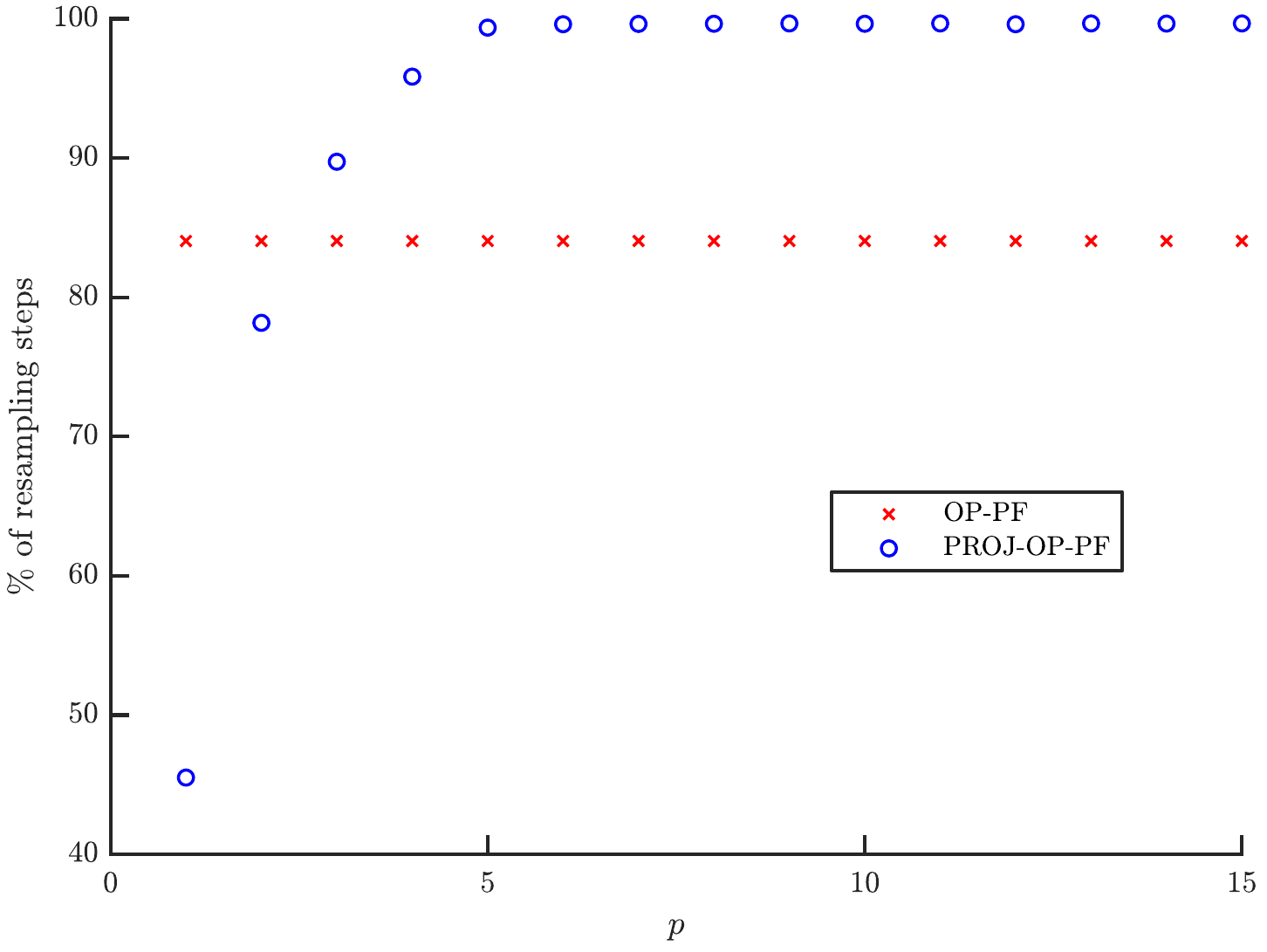}
\label{fig:bdegen}}
\caption{Statistics for PROJ-OP-PF as the rank of the projection is varied, compared to the Optimal Proposal PF, for the Lorenz 96 system with the standard $0.05$ time units between observations. Each data point represents the mean from 30 repetitions, each of which was also time-averaged. The optimal, $p=10$ PROJ-OP-PF RMSE is $2/3$ of the OP-PF RMSE. }
\label{fig:l96_bOP}
\end{figure}

The RMSE over time from one of the data points in Figure~\ref{fig:l96_bOP} is shown in Figure~\ref{fig:real}, and the DA methods are both shown over a long-time run. These clarify that the better performance of PROJ-OP-PF is not because it outperforms OP-PF at every, or even most data points. Rather, PROJ-OP-PF suffers from fewer spikes in the RMSE, and those spikes tend to be smaller.

\begin{figure*}
    \centering
    \includegraphics[width=0.4\textwidth]{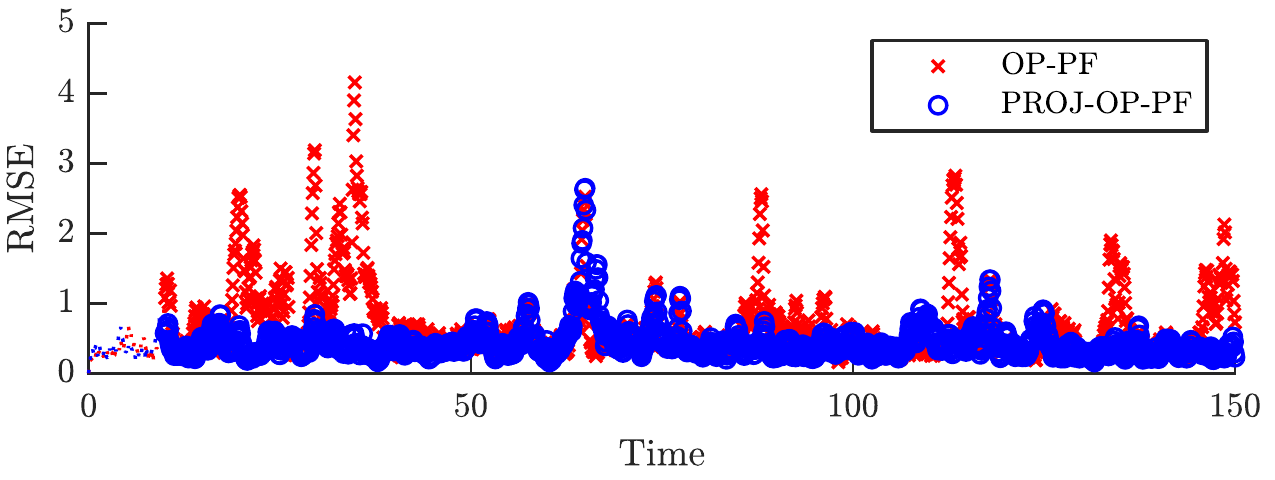}
    \caption{\jm{Error statistics for a long time run using the optimal $p=10$ parameters from Figure~\ref{fig:l96_bOP}. The spin-up steps are dotted. The short term errors (for the 200 steps after spinup) are 0.41 for PROJ-OP-PF and 0.64 for OP-PF. The long time errors are 0.43 for PROJ-OP-PF and 0.68 for OP-PF, suggesting the spinup is sufficiently long for error statistics to settle.} }
    % \includegraphics[width=0.43\textwidth]{pfig_L96_real_short.pdf}
    % \includegraphics[width=0.43\textwidth]{pfig_L96_real_long.pdf}
    % \caption{{\emph{Left:} one of the 30 runs of Figure~\ref{fig:l96_bOP} with $p=10$. The spinup time, during which errors are not measured, is dotted. RMSE: $0.53$ for PROJ-OP-PF and $0.75$ for OP-PF. It is coincidental that OP-PF only has error spikes starting at time 10. \emph{Right:} a long-time run, that displays similar error statistics to the shorter runs used in Figure~\ref{fig:l96_bOP}. RMSE: $0.39$ for PROJ-OP-PF and $0.49$ for OP-PF.}}
    \label{fig:real}
\end{figure*}

%In the previous section, PROJ-OP-PF was optimally tuned using resampling noise $\omega$ an order of magnitude larger than for OP-PF. This distinction is even more pronounced for Figures~\ref{fig:l96_bOP} and \ref{fig:real}. 
The results in Figure~\ref{fig:l96_bOP} were produced using $\omega = 0.0037$ for OP-PF, and up to $\omega=0.2$ for PROJ-OP-PF. The optimal value of $\omega$ was selected by computing the time-averaged RMSE for 30 repetitions of OP-PF at 25 different values of $\omega$ in $[10^{-4},\, 0.4]$. The tuning for PROJ-OP-PF additionally considers five values of $\alpha$ in $[0,\,1]$. The RMSE at each $(\omega,\,\alpha)$ are shown for the case when the projected data dimension is $p=10$ in Figure~\ref{fig:l96_tune_r_05}. The optimal values of $\omega$ for each $p$ are given in Table~\ref{tab:opo}. All these choices were optimal in the sense that they minimised the mean RMSE; one could instead, or additionally, have considered the prevalence of resampling in PROJ-OP-PF and tuned $(\alpha,\,\omega)$ to minimise that. \\

%{
The displayed RMSE in Figure~\ref{fig:l96_tune_r_05} is large whenever $\alpha=0$, and the same was true in Figure~\ref{fig:l96_tune_r}. One might infer that of the two ways the projection is used, in the weight update \eqref{projopW} and in PROJ-RESAMP (Algorithm~\ref{alg:r}), the latter is more significant. But in fact, if we run an OP-PF using the PROJ-RESAMP algorithm for resampling, we observe error statistics no better than the standard OP-PF. That is, both PROJ-OP-PF and PROJ-RESAMP are needed in concert to reliably improve on the Optimal Proposal PF.
%}

\begin{figure}
    \centering
    \includegraphics[width=0.16\textwidth]{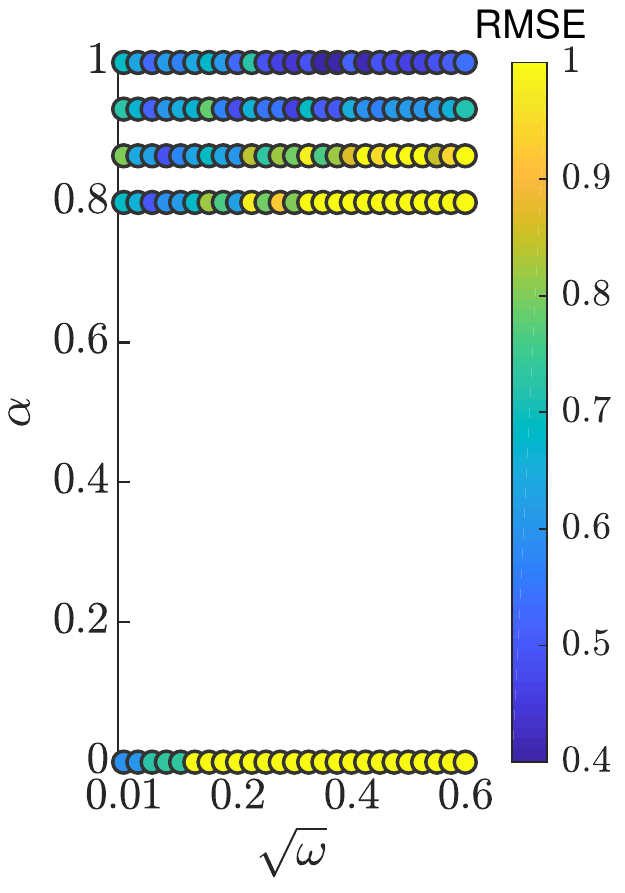} \includegraphics[width=0.16\textwidth]{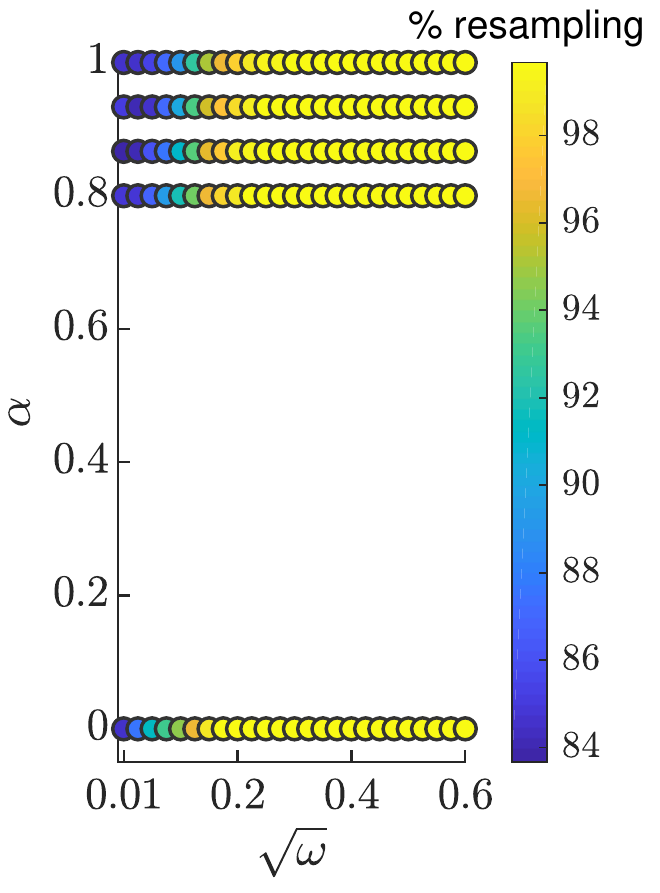}
    \includegraphics[width=0.16\textwidth]{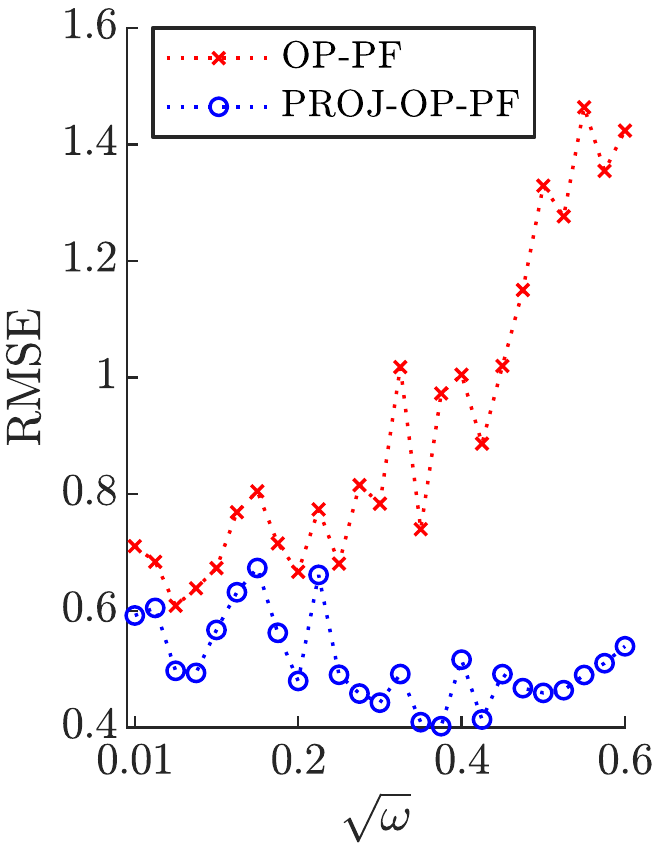}
\caption{\emph{Left, Middle}: Tuning results for PROJ-OP-PF with $p=10$ in Figure~\ref{fig:l96_bOP}. Data points are shaded to reflect the RMSE and percentage of resampling steps respectively, as the resampling noise $\omega$ and confinement to the unstable subspace $\alpha$ are varied for the Lorenz 96 system with time $0.05$ between observations. The RMSE ranges from $0.4$ to $2.1$, but is cut off at 1. \emph{Right}: Tuning results for OP-PF in Figure~\ref{fig:l96_bOP}. The optimal choice of RMSE, for $\omega = 0.0037$, is $0.61$. We also plot PROJ-OP-PF results with the optimal choice of $\alpha$ from the left two figures. Each data point in all figures represents the mean from 30 repetitions, each of which was also time-averaged. }
\label{fig:l96_tune_r_05}
\end{figure}

\begin{table}
    \centering
    \resizebox{0.49\textwidth}{!}{%
    \begin{tabular}{c|c c c c c c c c c c c c c c c}
         $p$ & 1 & 2  & 3 & 4 & 5 & 6 & 7 & 8 & 9 & 10 & 11 & 12 & 13 & 14 & 15 \\
         $\omega$ & 0.01 & 0.13  & 0.13 & 0.13 & 0.20 & 0.20 & 0.20 & 0.20 & 0.15 & 0.13 & 0.20 & 0.07 & 0.13 & 0.07 & 0.13
    \end{tabular}}
    \caption{Optimal choice of resampling noise $\omega$ for PROJ-OP-PF at each $p$ in Figure~\ref{fig:l96_bOP}. The optimal choice of $\alpha$ was $\alpha=1$ in all cases. }
    \label{tab:opo}
\end{table}

\subsubsection{Infrequent, accurate observations with small ensemble and high-dimensional model}
Finally, we investigate the behaviour of PROJ-OP-PF with a model dimension $J=400$. We use accurate model covariance $\Sc=0.04^2\I$ and observation covariance $\Rc=0.01^2\I$, and set the time between observations to $0.05$. \jm{The number of particles is $L=50$.} We choose to project onto the $p=8$ most unstable modes for PROJ-OP-PF, and use as a benchmark results from an ETKF. Results for this scenario are displayed in Figure~\ref{fig:highd}. We see the OP-PF diverge, while PROJ-OP-PF performs almost as well as the Ensemble Kalman Filter.

\begin{figure}
    \centering
    \includegraphics[width=0.47\textwidth]{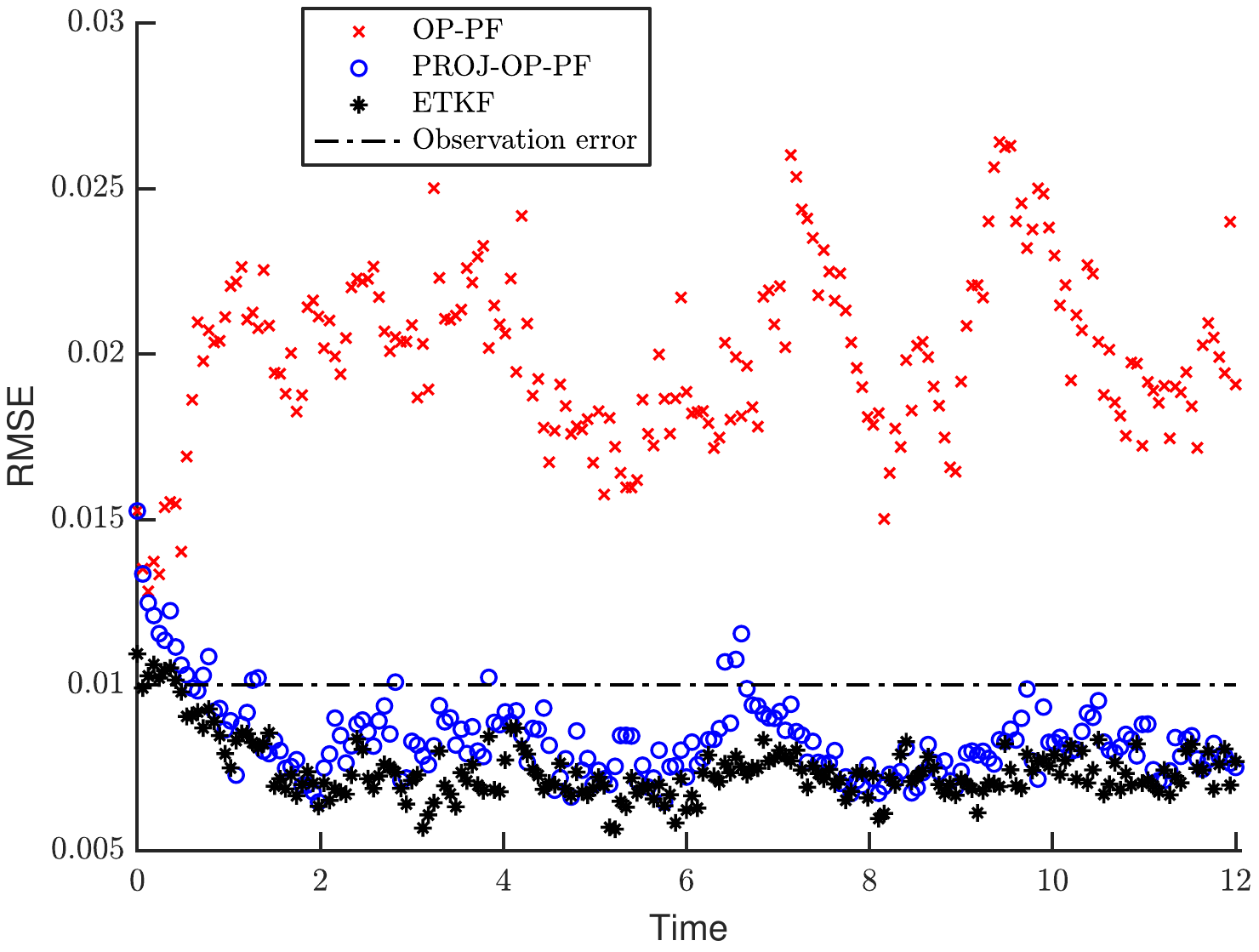}
    \caption{Error statistics for the DA methods over time, from the $400$-dimensional Lorenz96 system with accurate observations of every second variable. In this case the ETKF is used to provide a `good performance' benchmark for PROJ-OP-PF.}
    \label{fig:highd}
\end{figure}

%%%%%%%%%%%%%%%%%%%%%%%%%%%%%%%%%%%%%%%%%%%%%%%%%%%%%%%%%%%%%%%%%%%%%%%
\section{Discussion} \label{sec:disc}
In this work a new approach to DA has been derived that allows for dimension reduction of the data using a projection defined in state space. The chief application has been Particle Filters Assimilating in the Unstable Subspace, which the classical AUS approach is unsuitable for because ensemble methods already project the forecast strongly into the unstable subspace \citep{BocquetCarrassi17}. By contrast the new approach sharply reduces filter degeneracy in a predictable fashion, improves filter accuracy and allows one to construct a sensible resampling scheme that adds more noise in more uncertain directions. Algorithms resting on the projected DA approach were tested on a sample linear system to investigate the role of data dimension in a simple context, and on the chaotic Lorenz 96 system that provides a challenging scenario for particle filters. The projected DA approach was also found to have some benefits for the Ensemble Kalman Filter.
Two algorithms were tested; the first allows the projected DA formulation to be simply applied to any DA scheme, while the second is a particle filter that mixes projected and unprojected data based on the optimal proposal. The discrete QR technique used to find the unstable subspace in this work is rigorously justified and the additional cost incurred by it is proportional to employing an ensemble size of the dimension of the projected subspace. 

{Some limitations of the current projected algorithms suggest improvements that will drive further work in this area. The projected schemes make no use of the orthogonal data set, but in principle the orthogonal data could instead be assimilated in a separate algorithm that is less sensitive to dimension. Such manipulations are done in \cite{MajdaQiSapsis14, Slivinski15}, for example, and formulated for model error in AUS in section 3.2 of \cite{Grud2018a}. Future work will generalise the projected DA approach to the assimilation of multiple projections by multiple assimilation methods. }

\subsection*{Acknowledgements}
JM acknowledges the support of ONR grant N00014-18-1-2204, NSF grant DMS-1722578, and the Australian Research Council Discovery Project DP180100050. EVV acknowledges the support of NSF grants DMS-1714195 and DMS-1722578.
The authors are grateful to Alberto Carrassi for helpful feedback on an early version of this work.
\appendix
\section{Projections onto convex sets}
\label{dpa}
Given two orthogonal projections $\P_A,\,\P_B$, the following algorithms identify the projection $\P_{A\cap B}$.\\
Von Neumann's algorithm iterates the product of the projections,
\[
\P_{A\cap B} = \lim_{k\to\infty} (\P_A \P_B)^k
\]

Dykstra's projection algorithm generally converges faster.

Start with $x_{0}=\I,\, p_{0}=q_{0}=k=0$, and update by
\begin{align*}
 y_{k}=&\P_{A}(x_{k}+p_{k}) \\
 p_{{k+1}}=&x_{k}+p_{k}-y_{k}\\
 x_{{k+1}}=&\P_{B}(y_{k}+q_{k})\\
 q_{{k+1}}=&y_{k}+q_{k}-x_{{k+1}}.
 \end{align*}
Then $\P_{A\cap B} = \lim_{k\to\infty}x_k$. \\
Either algorithm may be used with some tolerance on the change in the approximation of $\P_{A\cap B}$, or to some finite $k$.

% \section{Projected state space and data models}
% \label{models}

% {\bf State Space Model} $u_{n+1} = F_n(u_n) + \xi_n$, {where $\xi_n\sim\N(0,\Sigma_u)$}

% {\bf Data Model} $y_{n+1} = H u_{n+1} + \eta_{n+1}$, {where $\eta_n\sim\N(0,\Sigma_y)$}

% {\bf State Space Model (perturbed form)} $u_{n+1}^{(0)}+\delta_{n+1} =
% F_n(u_n^{(0)}+\delta_n) + \xi_n$

% {\bf Data Model (perturbed form)} $y_{n+1} = H(u_{n+1}^{(0)}+\delta_{n+1}) + \eta_{n+1}$\\

% {\bf Projected State Space Model}
% $$\P_{n+1} \delta_{n+1} = \P_{n+1}[F_n(u_n^{(0)}+\P_n\delta_n)+\xi_n - u_{n+1}^{(0)}]$$
% {\bf Projected State Space Model (non perturbed form, non projected model integration)}
% \begin{align}
% \label{p.state} \P_{n+1} u_{n+1} = \P_{n+1}{F_n(u_n)}+\P_{n+1}\xi_n
%         \end{align}
% {\bf Projected Data Model}
% $$\P_{n+1}^\H  \Hdag y_{n+1} = \P_{n+1}^\H (u_{n+1}^{(0)}+\delta_{n+1}) + \P_{n+1}^\H  \Hdag \eta_{n+1}$$
% {\bf Projected Data Model (non perturbed form)}
% \begin{align}
% \label{p.data} \P_{n+1}^\H  \Hdag y_{n+1} = \P_{n+1}^\H  u_{n+1} + \P_{n+1}^\H  \Hdag \eta_{n+1}
% \end{align}

% where $\Hdag = H^T (HH^T)^{-1}$ [the pseudo inverse of $H$], and
% $\P_{n+1}^\H $ is the orthogonal projection onto the intersection of the subspaces spanned by the columns of $\U_{n+1}$ and the rows of $H$.
% Note that $\P_\H = H^T(HH^T)^{-1}H = \Hdag H$ (we are assuming $H$ full rank).
% $\P_{n+1}\delta_{n+1}$ and $\P_{n+1}^\H  \Hdag y_{n+1}$ can be
% written as a linear combination of the $p<m$ columns of $\U_{n+1}$.

% \bibliographystyle{plain}
% \bibliographystyle{siamplain}
\bibliography{bibliography}{}
\bibliographystyle{wileyqj}

\end{document}